\documentclass[12pt]{arxiv_version}

\usepackage{mathtools}
\usepackage{mathrsfs}
\usepackage{stmaryrd}
\usepackage{bm}
\usepackage{enumitem}
\usepackage{booktabs}
\usepackage{array}
\usepackage{longtable}
\usepackage{adjustbox}

\usepackage{algorithm}
\usepackage{algpseudocode}

\usepackage{graphicx}
\usepackage{subcaption}
\usepackage{tikz}
\usepackage{multirow}
\definecolor{deepjunglegreen}{rgb}{0.0, 0.29, 0.29}

\definecolor{mygreen}{RGB}{0, 100, 0}

\usepackage{cleveref}

\usepackage[most]{tcolorbox}

\DeclareMathOperator*{\argmax}{arg\,max}

\newcommand{\mSet}[1]{\{#1\}}

\NewDocumentCommand{\randactiv}{}{{\boldsymbol{\sigma}}}

\NewDocumentCommand{\randEdge}{}{{\boldsymbol{E}}}

\NewDocumentCommand{\randcirc}{}{{\boldsymbol{\hat{f}}}}

\NewDocumentCommand{\Bool}{}{{\mathbb{G}^{\text{Bool}}}}
\NewDocumentCommand{\Bary}{}{{\mathbb{G}^{\star}}}

\newcommand{\eqdef}{\mathrel{\stackrel{\scriptscriptstyle\mathrm{def}}{=}}}

\newcommand{\R}{\mathbb{R}}
\newcommand{\N}{\mathbb{N}}

\newcommand{\AND}{\texttt{AND}}

\definecolor{junglegreen}{rgb}{0.16, 0.67, 0.53}

\NewDocumentCommand{\AK}{moo}{
    \IfValueF{#2}{
        \IfValueF{#3}{%
                        {{%
                            \textcolor{deepjunglegreen}{%
                                \textbf{Annie:}%
                                \textit{{#1}}%
                            }%
                        }}%
        }%
    }
    \IfValueT{#2}{
        \IfValueF{#3}{
                        {{\scriptsize
                            \textcolor{deepjunglegreen}{
                                \hfill\\
                                    \textbf{Annie:}
                                    \textit{{#1}}
                                \hfill\\
                            }
                        }}
        }
        \IfValueT{#3}{
                        \marginnote{{\scriptsize
                            \textcolor{deepjunglegreen}{ 
                            \textbf{Annie:}
                            \textit{{#1}}
                            }
                        }}
                }
        }
                    }

\newcommand{\ModelName}{Stochastic Boolean Circuit}
\newcommand{\ModelNameShort}{SBC}

\title{Certifiable Boolean Reasoning Is Universal}

\coltauthor{
    \Name{Wenhao Li}\thanks{Equal contribution. Corresponding authors.}
    \Email{chriswenhao.li@mail.utoronto.ca}\\
    \addr Department of Mechanical \& Industrial Engineering, University of Toronto, Canada; 
    \AND 
    \Name{Anastasis Kratsios}\footnotemark[1]
    \Email{kratsioa@mcmaster.ca}\\
    \addr Department of Mathematics, McMaster University and Vector Institute, Canada; 
     \\
     The Ennio De Giorgi Mathematical Research Centre, Scuola Normale Superiore di Pisa, Italy.
    \AND           
    \Name{Hrad Ghoukasian} \Email{ghoukash@mcmaster.ca}\\
    \addr Department of Mathematics, McMaster University and Vector Institute, Canada; 
    \AND
    \Name{Dennis Zvigelsky} \Email{dennis.zvigelsky@mcmaster.ca}\\
    \addr Department of Mathematics, McMaster University and Vector Institute, Canada; 
}

\begin{document}
\maketitle

\begin{abstract}
The proliferation of agentic systems has thrust the reasoning capabilities of AI into the forefront of contemporary machine learning. 
While it is known that there \emph{exist} neural networks which can reason through any Boolean task $f:\{0,1\}^B\to\{0,1\}$, in the sense  that they emulate Boolean circuits with fan-in $2$ and fan-out $1$ gates, trained models have been repeatedly demonstrated to fall short of these theoretical ideals.
This raises the question: \textit{Can one exhibit a deep learning model which \textbf{certifiably} always reasons and can \textbf{universally} reason through any Boolean task?} Moreover, such a model should ideally require few parameters to solve simple Boolean tasks.

We answer this question affirmatively by exhibiting a deep learning architecture which parameterizes distributions over Boolean circuits with the guarantee that, for every parameter configuration, a sample is almost surely a valid Boolean circuit (and hence admits an intrinsic circuit-level certificate).  We then prove a universality theorem: for any Boolean $f:\{0,1\}^B\to\{0,1\}$, there exists a parameter configuration under which the sampled circuit computes $f$ with arbitrarily high probability.  When $f$ is an $\mathcal{O}(\log B)$-junta, the required number of parameters scales linearly with the input dimension $B$. 

Empirically, on a controlled truth-table completion benchmark aligned with our setting, the proposed architecture trains reliably and achieves high exact-match accuracy while preserving the predicted structure: every internal unit is Boolean-valued on $\{0,1\}^B$. Matched MLP baselines reach comparable accuracy, but only about $10\%$ of hidden units admit a Boolean representation; i.e.\ are two-valued over the Boolean cube.
\end{abstract}


\section{Introduction}
\label{s:Intro}

With the recent boom in 
agentic models that act over extended horizons, e.g.~\cite{shinn2023reflexion,wang2024voyager}, 
AI-based reasoning systems that orchestrate tools and structured search, e.g.~\cite{shen2023hugginggpt,schick2023toolformer,yao2023treeofthoughts,yao2023react}, and 
iterative self-improvement techniques for refining solutions or bootstrapping rationales, e.g.~\cite{madaan2023self,wang2023selfconsistency,zelikman2022star,chen2025magicore,song2025thinking},
(purported) AI reasoning has become commonplace in modern machine-learning pipelines, increasingly influencing end-user decisions.
This burst of AI reasoning models rests on well-forged connections between deep learning and complexity theory from: the Turing machine lens~\cite{siegelmann2012neural,perez2018on,chung2021turing,mali2021neural,bournez2025universal} for recursive/looped networks; the Boolean~\cite{maass1991sigmoidvsbool,siu1993division,merrill2022saturated,merrill2023parallelism,chiang2024uniformtc0} and general circuit complexity lenses~\cite{lowd2008learning,kratsios2025quantifying} for non-recursive architectures (e.g., MLPs and transformers); as well as connections to mathematical logic~\cite{blumer1989learnability,karpinski1997polynomial,margulies2016polynomial,cook_nguyen_2010_logical_foundations}. Broadly speaking, these results show that there \textbf{exist} deep learning models which can, in principle, be configured to reason through arbitrary tasks; 
when reasoning is understood as
exactly emulating
an algorithm, or equivalently a \textit{circuit}, computing some (Boolean) function.  

In contrast, there is mounting evidence that most real-world \textit{trained} AIs \textit{fail} to reason through complex multi-step reasoning tasks~\cite{keysers2020measuring,wei2022chain,zhao2024mathtrap}, arithmetic~\citep{cobbe2021training} or more sophisticated mathematical reasoning~\citep{hendrycks2021measuring}, elementary deductive logic~\cite{berglund2023reversal}, inductive reasoning~\cite{olsson2022induction,zhao2024mathtrap}, and several other forms of illusory reasoning peacocking~\cite{song2024survey}. 
Experimental evidence increasingly indicates that trained deep learning pipelines produce effective memorizers of exorbitantly many training instances~\cite{bender2021stochastic,wu2024reasoning,shojaee2025the}.

Even seemingly innocuous training-data contamination~\cite{zou2025poisonedrag,chaudhari2026exploiting} or adversarial prompt attacks~\cite{liu2024formalizing,rajeev2025cats} are known to devastate these models’ reasoning.  This heaps on the evidence that AIs
\emph{trained} on \emph{data} often \emph{fail} to reason and, even when they do, their chain-of-thought can be broken by adversarial prompt injection.  Thus, there is very real \textbf{AI safety concern}; namely, how can real-world AIs (not those configured within a mathematical vacuum e.g.~\cite{kratsios2025quantifying}) be guaranteed to certifiably reason.
This leads to a natural question:
\begin{tcolorbox}
How can one construct a neural network architecture which always admits a \emph{certifiable}, \emph{universal}, and \emph{efficient} circuit-level interpretation, while remaining \emph{practically competitive} on modern propositional reasoning benchmarks?
\end{tcolorbox}
More precisely, we seek an architecture satisfying:
\begin{enumerate}
    \item[(i)] \textbf{Certifiable Reasoning:} For \emph{every} choice of network parameters, the model induces a meaningful distribution over ``simple and interpretable'' Boolean circuits, i.e. here, Boolean circuits whose gates have fan-in $2$ and fan-out $1$. In particular, any draw from this distribution yields, almost surely, a valid and interpretable Boolean circuit.
    \item[(ii)] \textbf{Universality:} With sufficiently many neurons, the parameters \emph{can be} configured so that the model samples a circuit of the form~(i), computing
    any Boolean function on $B$ bits; with arbitrarily high probability.
    \item[(iii)] \textbf{Efficiency:} If $f:\{0,1\}^B\to \{0,1\}$ is sparse,
    then achieving~(ii) requires only $\mathcal{O}(B)$ parameters.
    \item[(iv)] \textbf{Competitive Performance:} The model trains reliably and attains state-of-the-art performance on propositional logic tasks.
\end{enumerate}
Our main contribution is the construction of a simple neural network model which answers the theoretical desiderata~(i)--(iii) \textit{affirmatively}; while validating~(iv) experimentally.

\subsection{Contribution}
\label{s:Intro__ss:Contribution}
We answer this 
sensitive question in the foundations of reasoning-based AI safety \textbf{affirmatively} for Boolean tasks.
Our solution introduces trainable stochastic Boolean circuits with fan-in $2$ and fan-out $1$ gates; i.e.\ for any parametric configuration (outlined in Figure~\ref{fig:Model} and formally defined in~\eqref{eq:Model_Def}).  
We prove that our model samples a meaningful Boolean circuit almost surely (Theorem~\ref{thrm:StructureAware}).
We then show that one may select the model parameters so that a typically sampled circuit will compute any 
desired Boolean function with arbitrarily high probability, given enough neurons (Theorem~\ref{thrm:UniversalReasoning}), i.e., it is universal.
For sufficiently sparse Boolean functions, so-called juntas, the number of trainable parameters can be linear in the number of input bits.

All our results are verified experimentally: although multilayer perceptrons (MLPs) exhibit comparable reasoning performance when completing the truth tables of randomly sampled Boolean propositions, most (roughly $80\%$) of the neurons in these models are verifiably not computing any Boolean gate, casting doubt on their reasoning abilities (cf.~Definition~\ref{defn:BNR_Ratio}), whereas this cannot be the case for our AI.

\subsection{Organization of Paper}
\label{s:Intro__ss:Organization}
Section~\ref{s:Prelim} aggregates the notation, background, and terminology required to formulate and frame our results.
Section~\ref{s:Model} introduces our model and explains the theoretical principles motivating each architectural innovation.
Section~\ref{s:Theory} contains our main results; namely, a universal computation theorem (Theorem~\ref{thrm:UniversalReasoning}), its quantitative $\mathcal{O}(\log(B))$-junta version (Theorem~\ref{thrm:UniversalReasoning__Quantitative}), and the certified reasoning structure of our model, guaranteed for all parameter configurations (Theorem~\ref{thrm:StructureAware}).
Our theoretically principled architecture is then evaluated and compared against classical deep learning models experimentally in Section~\ref{s:Experiments}.
All proofs and additional experimental details are relegated to the appendices.

\section{Preliminaries}
\label{s:Prelim}

\subsection{Setting and Notation}
\label{s:Prelim__ss:Setting}
We henceforth fix a probability space $(\Omega,\mathcal{F},\mathbb{P})$ on which all random variables are defined.  All random variables will be denoted in boldface.
We write $\mathbb{N}$ for the set of non-negative integers, $\mathbb{N}_+$ for the set of positive integers and, for every $N\in \mathbb{N}$ we write $\mathbb{N}_N\eqdef \{n\in \mathbb{N}:\, n\ge N\}$, $[N]\eqdef \{n\in \mathbb{N}:\, n\le N\}$, and $[N]_+\eqdef [N]\cap [1,\infty)$.  

\noindent
For any $n,m\in \mathbb{N}_+$, each $i\in[n]_+$ and $j\in [m]_+$ and every $n\times m$ matrix $A$ we write $A_{i:}$ for the $i^{th}$ row of $A$ and $A_{:j}$ for its $j^{th}$ column.  All vectors are column vectors.  
For any matrix (resp.\ vector) $\|A\|_0$ counts the number of non-zero entries in $A$.
The set of \textit{all} binary Boolean gates of fan-in $2$ and fan-out $1$, cf.~Section~\ref{s:Prelim__ss:BooleanLogic}, is denoted by $\Bary$.

\noindent
Fixing $N\in \mathbb{N}_+$, $\mathbf{1}_N\in \mathbb{R}^N$ denotes the vector of $1$s.  We denote the $N$-simplex by $\Delta_N\eqdef \{w\in [0,1]^N:\, \mathbf{1}_N^{\top}w=1\}$, its relative interior is
$\dot{\Delta}_N \eqdef \{w\in(0,1)^N:\, \mathbf{1}_N^\top w = 1\}$.  We
define the \textit{softmax} function $\operatorname{SM}_N:\mathbb{R}^N\to \Delta_N$ as mapping any $x\in \mathbb{R}^N$ to $\operatorname{SM}_N(x)\eqdef \big(
{e^{x_n}}/{\sum_{m=1}^N\, e^{x_m}}\big)_{n=1}^N$.
When the dimension is clear from context, we simply write $\operatorname{SM}$.

\subsection{Boolean Circuits of Fan-in 2 and Fan-out 1, and the \texorpdfstring{$\operatorname{AC}^0$}{Alternating Constant Depth} Class}
\label{s:Prelim_AC}

\noindent
\textbf{Directed Acyclic Graphs}
A directed graph $G=(V,E)$ is a pair of a (non-empty) set $V$ of \textit{nodes} and a collection of ordered pairs
$E\subseteq V^2\setminus \{(v,v)\}_{v\in V}$ representing directed edges.  

Given two nodes $v,w\in V$, a directed path between $v$ and $w$ is a sequence
$(v,v_1)$, $(v_1,v_2)$, $\dots$, $(v_k,w)\in E$ ``connecting $v$ to $w$,'' and $k+1$ is the \textit{length} of this directed path.  
Such a $G$ is called \textit{connected} if there is a directed path between any distinct pair of nodes. $G$ is called \textit{acyclic} if there are no directed paths from a node to itself.  
If $G$ is connected and acyclic we call it a \textit{connected DAG}.    
Any two nodes $u,v\in V$ are partially ordered, written $u\lesssim v$,
cf.~\citep[Proposition~2.1.3]{bang2008digraphs}, 
if there is a directed path originating at $u$ and terminating at $v$.  
The cardinality of the largest set of unordered vertices in $G$ is the \textit{width} of $G$.  
The \textit{depth} of $G$ is the length of the longest directed path in $G$.\footnote{For the directed graph underlying the computational graph of any MLP these notions of width and depth coincide with the usual uses of the words in MLP.}

In contrast, an undirected graph $G_u = (V, E_u)$ is a pair of a (non-empty) set $V$ of \emph{nodes} and a collection of \emph{unordered} pairs $E_u \subseteq V^2 \setminus\mSet{(v, v)}_{v \in V}$ representing edges. An undirected graph is called a \emph{tree} if for every distinct $v_1,v_2\in V$ there exists a unique path from $v_1$ to $v_2$.
A \textit{DAG} is called a directed \textit{tree} if
its underlying undirected graph is a tree.
The set of nodes with no directed path terminating thereat are called \textit{roots/input nodes} and the set of nodes for which there is no directed path originating thereat are called \textit{leaves/output nodes}.  A binary tree is a tree in which every non-root node $v\in V$ has exactly two distinct nodes $u,w\in V$ such that $(u,v),(w,v)\in E$.

\noindent
\textbf{Boolean Circuits}
\textit{Boolean gate}s $g$ with fan-in $2$ and fan-out $1$ are ``simple'' elementary Boolean computations, i.e., Boolean functions compatible with the structure of a binary (directed) tree;
more precisely, $g:\{0,1\}^2\to \{0,1\}$; e.g.\ $\operatorname{AND}$, $\operatorname{OR}$, or the lifted not gate
$\operatorname{NOT}(x_1,x_2)\eqdef 1-x_1$.
There are exactly $16$ such gates; cf.~\ref{tab:all_possible_gates}
in Appendix~\ref{s:GateEnumeration}.

Fix $B,\Delta,\Upsilon\in\mathbb{N}_+$.
In this paper, a \textit{Boolean circuit} with $B$ input bits, depth $\Delta$, and width $\Upsilon$ refers to a triple
$\mathcal{C}=(V,E,g_{\cdot})$ of a DAG $G=(V,E)$ with
$B$ root/input nodes $v_{\operatorname{in}:1},\dots,v_{\operatorname{in}:B}\in V$ and
one terminal/output node $v_{\operatorname{out}}\in V$, {{such that: the in-degree of every
non-root node is exactly $2$; except for an initial \textit{bit-lifting channel} which selects bits and their negations.  Such a bit-lifting channel is formalized by a map $x\mapsto \Pi (x,\mathbf{1}_B-x)$ where $\Pi$ is a matrix with $1$-hot vectors as rows.
}}
The set of all nodes which are neither input nodes nor outputs of the bit-lifting channel
are called \textit{computational nodes} and are denoted by
$
V_{\operatorname{comp}}
\eqdef
V \setminus 
\{v_{\operatorname{in}:b}\}_{b=1}^B
$.
Additionally, $g_{\cdot}=(g_v)_{v\in V_{\operatorname{comp}}}$ 
is a family of Boolean gates of fan-in $2$ and fan-out $1$ .
The set of all such Boolean circuits%
\footnote{Any Boolean gate of fan-in $2$ and fan-out $1$ is allowable in this class; one may also consider the subclass where only $\operatorname{AND}$, $\operatorname{OR}$, and $\operatorname{NOT}$ gates are used, with no essential change to the theory.}~%
is denoted by $\operatorname{CIRCUIT}_{\Delta,\Upsilon}^B$.
Since $\mathcal{C}$ has depth $\Delta$, the vertex set admits a layering
$
V=\bigsqcup_{l=0}^{\Delta} V_l
$,
with $V_0=\{v_{\operatorname{in}:b}\}_{b=1}^B$, $V_\Delta=\{v_{\operatorname{out}}\}$, such that every directed edge $(u,v)\in E$ satisfies $u\in V_l$ and $v\in V_{l+1}$ for some $l\in\{0,1,\dots,\Delta-1\}$.
Furthermore, since $\mathcal{C}$ has fan-in $2$, every non-input vertex $v\in V_l$ with $l\ge 2$ has exactly two predecessors $v_{-1},v_{-2}\in V_{l-1}$ satisfying $(v_{-1},v),(v_{-2},v)\in E$.

\noindent
\textbf{Boolean Reasoning as Computation of Boolean Functions}
Applying the current AI reasoning theory lens~\cite{merrill2022saturated,chiang2024uniformtc0} we understand reasoning through a ``Boolean task'' to mean the computation of a Boolean function by a Boolean circuit in
$\cup_{\Delta,B,\Upsilon\in \mathbb{N}+} \operatorname{CIRCUIT}_{\Delta,\Upsilon}^B$; the latter being defined in the standard sense of circuit complexity theory, cf.~\cite{jukna2012boolean}, and resting on the classical correspondence between Boolean circuits and first-order propositional logic, cf.~\cite{cook_nguyen_2010_logical_foundations}.
Fix a choice of $B^{\uparrow}\in\mathbb{N}_+$ (with repetition allowed) input variables from $\{x_b\}_{b=1}^B$ and their negations $\{(1+x_b)\,\bmod{2}\}_{b=1}^B$, with $B^{\uparrow}\le \Upsilon$, equivalently encoded by a $B^{\uparrow}\times 2B$ row-stochastic matrix $\Pi$. 
Then every Boolean circuit $\mathcal{C}\in \operatorname{CIRCUIT}_{\Delta,\Upsilon}^B$ induces a Boolean function
$f_{\mathcal{C}}:\{0,1\}^B\to\{0,1\}$ defined recursively as follows:
\begin{equation}
\begin{aligned}\label{eq:recursive_defenition}
f_{\mathcal{C}}(x) &\eqdef f_\Delta(x)
\\
f_l(x) &\eqdef \bigl(g_v\bigl(
f_{l-1}(x)_{v_{-1}},f_{l-1}(x)_{v_{-2}}
\bigr)\bigr)_{v\in V_l},
\qquad l=2,\dots,\Delta,
\\
f_1(x) &\eqdef x_1
\\
x_1 &\eqdef \Pi 
\begin{pmatrix}x\\ \mathbf{1}_B-x\end{pmatrix},
\end{aligned}
\end{equation}
where $v_{-1}$ and $v_{-2}$ refer to the indices on the vector $f_{l-1}(x)$ corresponding the two parent nodes of $v$.  
We say that a Boolean function $f:\{0,1\}^B\to \{0,1\}$ is \textit{computable} if there exists a Boolean circuit $\mathcal{C}$ such that its induced function $f_{\mathcal{C}}$ satisfies
\[
f_{\mathcal{C}}(x)=f(x)\qquad \text{for every }x\in \{0,1\}^B.
\]
There are several \textit{complexity classes} in the theory of non-uniform models of computing/circuit complexity, e.g.\ $\operatorname{AC}^k$, $\operatorname{NC}^k$, $\operatorname{TC}^k$, etc., for $k\in \mathbb{N}$.
These do not describe which Boolean functions are, or are not, computable, but rather classify how hard it is to compute a given computable Boolean function by bounding the depth of the circuits computing it, depending on how ``powerful'' the allowed gates are.
For instance, $\operatorname{AC}^k$ allows unbounded fan-in $\operatorname{AND}$, $\operatorname{OR}$, and $\operatorname{NOT}$ gates; $\operatorname{NC}^k$ restricts to bounded fan-in $\operatorname{AND}/\operatorname{OR}/\operatorname{NOT}$ gates, corresponding to highly parallelizable circuit families; and $\operatorname{TC}^k$ allows the use of more powerful gates such as threshold (majority) gates, which are not computable by more elementary Boolean gate sets; cf.\ e.g.~\cite{Hastad86} in polynomial time.
\hfill\\
A key new set of properties which our model manages to ensure simultaneously is that it almost surely samples a random element of $\operatorname{CIRCUIT}_{\Delta,\Upsilon}^B$ for prescribed $\Delta$, $\Upsilon$, and $B$ (Theorem~\ref{thrm:StructureAware}), and that every such circuit is sampleable by our universal model (Theorem~\ref{thrm:UniversalReasoning}).
Moreover, we may train its weights and biases by differentiable optimization so as to guide it toward any given target circuit (Section~\ref{s:Experiments}).

\subsection{A Test for Latent Boolean Reasoning}
\label{s:BnR}
Before moving onto we ask ourselves: \textit{``What can it mean for a non-Boolean function to latently be performing Boolean reasoning?''} 
If we can formalize this, then we can easily test whether neurons in standard neural networks, e.g.\ MLPs, are performing a latent form of Boolean reasoning.  Whatever our definition may be, it should be aligned with the ideas behind reasoning probes, e.g.~\citep{Alain2016linearprobes,craven1996trepan,dai-etal-2022-knowledge}, and it should always be satisfied for our idealized reasoners; namely, Boolean circuits.

We understand a function $f$ from a Boolean cube $\mSet{0, 1}^B$ with values in some Euclidean space as admitting a ``Boolean neuronal representation", or interpretation, if each component of $f$ can be understood as assigning Boolean values to the inputs of $x$; thus, each component can be interpreted as a neuron.  For this to be possible, of course, each potential neuron of $f$, i.e.\ each component thereof, cannot assign more than two distinct values to the members of the Boolean cube; otherwise there have to be more than two truth values which is impossible.
\begin{definition}[Boolean Neuronal Representability]
\label{defn:BnB}
Let $B_{in},B_{out}\in \mathbb{N}_+$.  A function $f:\{0,1\}^{B_{in}}\to \mathbb{R}^{B_{out}}$ is said to admit a Boolean neuronal representation if: 
for every $i\in [B_{out}]_+$ there is a map $\mathcal{R}_i: \mathcal{X}_i\to \{0,1\}$ which is injective; where $\mathcal{X}_i\eqdef  p_i\circ f(\{0,1\}^{B_{in}})$ and $p_i$ denotes the $i^{\text{th}}$ coordinate projection.
Whenever such $\mathcal{R}_1,\dots,\mathcal{R}_{B_{out}}$ exist, $f$ is said to be
\textit{Boolean neuronally representable (BNR)}, and any map
$
\mathcal{R}:\mathbb{R}^{B_{out}}\to\{0,1\}^{B_{out}}
$
is called a Boolean neuronal representation of $f$, if
\[
\mathcal{R}(f(x))=\big(\mathcal{R}_i(f(x)_i)\big)_{i=1}^{B_{out}}
\qquad\mbox{for all } 
x\in\{0,1\}^{B_{in}}
.
\]
\end{definition}
By construction, every Boolean function is BNR.  More generally, $\{0,1\}$-valued functions such as MLPs with the neuroscience-inspired  \textit{threshold activation function}, e.g.~\cite{mcculloch1943logical}, also achieve a BNR density of $1$ (see Definition~\ref{defn:BNR_Ratio}).  
In contrast, even simple modern neurons, e.g.\ ReLU neurons, easily fail to be BNR (see Example~\ref{ex:BnR} and Proposition~\ref{thm:relu-not-bnr-w.h.p.}).

\section{The {\ModelName} Model}
\label{s:Model}
Our \ModelName\ model (Figure~\ref{fig:Model}) is assembled from three components: a parameterized distribution over fan-in $2$, fan-out $1$ Boolean gates; a distribution over pairs of edges; and a stochastic bit-lifting channel.

\begin{figure}[ht!]
    \centering
    \begin{minipage}[t]{0.49\linewidth}
        \centering
        \includegraphics[width=.35\linewidth]{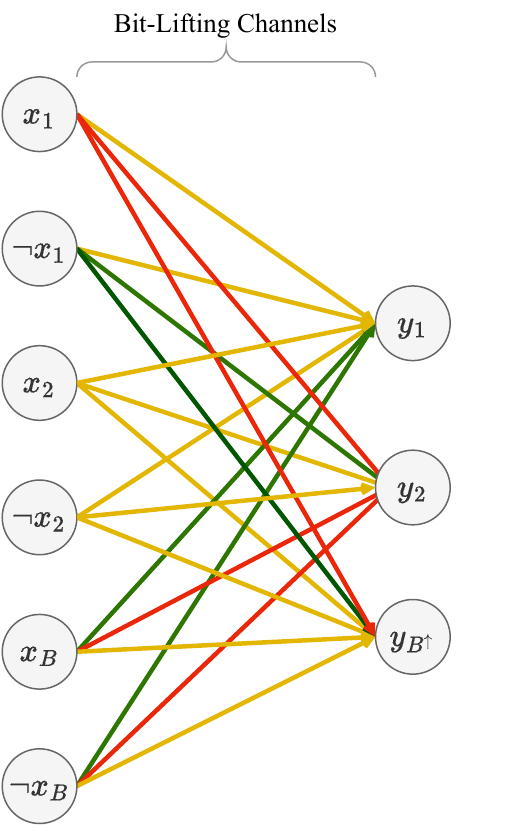}
        \caption*{(a) \textit{Stochastic Bit-Lifting Channels}:
        The first layer of our model employs stochastic bit-lifting channels (cf.~\eqref{eq:lifted_guy}), which select $B^{\uparrow}$ among the $B$ input bits and their negations (duplicates allowed).}
        \label{fig:Model__Bitliftingchannel}
    \end{minipage}\hfill
    \begin{minipage}[t]{0.49\linewidth}
        \centering
        \includegraphics[width=.4\linewidth]{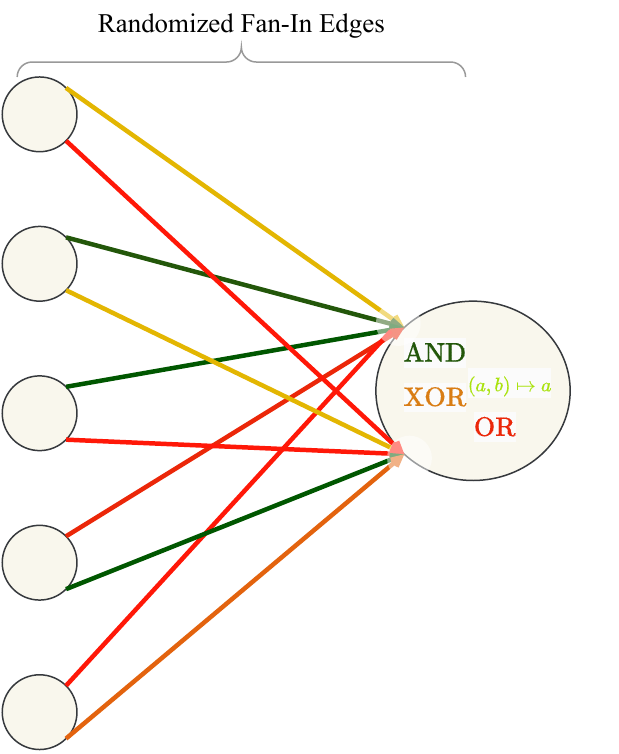}
        \caption*{(b) \textit{Randomized Boolean Neuron}:
        Exactly two edges (fan-in $2$) are selected from the nodes of the previous layer (cf.~\eqref{eq:NoCrossing_EdgeWeights}); then, a Boolean gate is randomly selected (cf.~\eqref{eq:EveryPossibleGateAllAtOnce}) and the output is computed.}
        \label{fig:Model__Computation}
    \end{minipage}
    \caption{Building blocks of the {\ModelName} model:
    Our feedforward model defines a distribution over Boolean circuits (Theorem~\ref{thrm:StructureAware}) whose parameters can be chosen so that it samples any desired circuit with arbitrarily large probability (Theorem~\ref{thrm:UniversalReasoning}).
    \hfill\\
    \textbf{Phase 1.} The \ModelNameShort\ uses stochastic bit-lifting channels to select the most relevant $B^{\uparrow}$ bits and their negations to be passed to the rest of the network.
    \hfill\\
    \textbf{Phase 2.} Each node in each subsequent layer randomly selects exactly two inputs and applies a random fan-in $2$, fan-out $1$ Boolean computation.
    \hfill\\
    \textbf{What Parameters Control:} The {\color{red}{red}} edges (resp.\ gates) illustrate edges which are less likely to be selected, and the {\color{green}{green}} edges (resp.\ gates) are more likely to be drawn. The probability of this selection occurring is determined by the trainable weight parameters of the bit-lifting channels and the neurons.
    }
    \label{fig:Model}
\end{figure}

We now rigorously define the components of the \ModelName\ model.

\subsection{Model Components: Stochastic Bit-Lifting Channels}
\label{s:LiftingChannels}

Let $B,B^{\uparrow}\in \mathbb{N}_+$, and let $W\in \mathbb{R}^{B^{\uparrow}\times 2B}$.  
Defining the stochastic bit-lifting channel $\mathcal{L}(\cdot|W):\mathbb{R}^B\to \mathbb{R}^{B^{\uparrow}}$ for every $x\in \mathbb{R}^B$ by
\begin{equation}
\label{eq:bit_lifting__channel}
        \mathcal{L}(x|W)
    \eqdef
        \operatorname{SM}(W)\,
        \left(
            \begin{pmatrix}
                I_B\\
                -I_B
            \end{pmatrix}
            x
            +
            \begin{pmatrix}
                \mathbf{0}_B\\
                \mathbf{1}_B
            \end{pmatrix}
        \right)
\end{equation}
where the softmax function $\operatorname{SM}$ is applied row-wise to the matrix $W$.  
We use $\mathcal{L}(\cdot|W)$ to parametrize the distribution of the following random variable.
Specifically, let $\mathbf{X}^{\uparrow:W}$ be a $C(\{0,1\}^B,
\{0,1\}^{\uparrow B})$-valued random variable whose $i^{th}$ coordinate $\mathbf{X}^{\uparrow:W}_i$ has distribution, given for any $x\in \{0,1\}^B$, by 
\begin{equation}
\label{eq:lifted_guy}
        \mathbf{X}^{\uparrow:W}_i(x)
    \sim
            \sum_{j=1}^B
                \big(\operatorname{SM}_{2B}(W_{i:})\big)_j
                \,
                \delta_{x_j}
        +
            \sum_{j=1}^{B}
                \big(\operatorname{SM}_{2B}(W_{i:})\big)_{j+B}
                \,
                \delta_{\neg x_j}.
\end{equation}

\subsection{Model Components: Randomized Boolean Neuron}
\label{s:Model__Neuron}

\subsubsection{Edges: Trainable Random Fan-in Edges}
\label{s:Model__ss:UncrossableEdges}
For every $B\in \mathbb{N}_+$ and every temperature parameter $\eta>0$, define the operation $\operatorname{S}_{\eta}:\mathbb{R}^B\times \mathbb{R}^B\to \mathbb{R}^B\times \mathbb{R}^B$ as mapping any $(W_1,W_2)\in \mathbb{R}^B\times \mathbb{R}^B$ to
\begin{equation}
\label{eq:NoCrossing_EdgeWeights}
        \operatorname{S}_{\eta}\big(
            W_1
            ,
            W_2
        \big)
    \eqdef 
        \Big(
            \operatorname{SM}(\eta W_1)
        ,
                \operatorname{SM}\big(
                    \eta(1-\operatorname{SM}(\eta W_1))
            \odot
                \operatorname{SM}(\eta W_2) \big)
        \Big)
\end{equation}
where $\odot$ denotes the componentwise (Hadamard) product on pairs of vectors in $\mathbb{R}^B$.
For $i = \{1, 2\}$, $S_\eta(W_1, W_2)_i$ represents the distribution in $\dot{\Delta}_B$ for input $i$ of a given Boolean fan-in 2 gate. 
\subsubsection{Gates: Trainable Stochastic Boolean Gates}
\label{s:Model__Neuron___GateSelector}
    Let
$\bar{x}_1,\ldots,\bar{x}_4$ be the elements of $\{0,1\}^2$ ordered as
$(0,0),(0,1),(1,0),(1,1)$. For every truth assignment $Z \in \{0,1\}^4$, define the quadratic interpolant  (note differentiable) 
$\tilde{\sigma} : \{0,1\}^2 \to \{0,1\}$ for each $x \in \{0,1\}^2$ by
    \begin{align}
    \label{eq:basic_activation}
         \tilde{\sigma}
         (x|Z) 
        \eqdef 
            \sum_{i=1}^4\,
                \varsigma(x-\bar{x}_i)Z_i,
    \quad\mbox{ where }\quad
            \varsigma(x)
        \eqdef
            \begin{cases}
                c
                \,
                e^{1/(4\|x\|_2^2-1)}
                & \mbox{ if } 
                \|x\|_2 < \tfrac{1}{2}
            \\
                0 & \mbox{ if } \|x\|_2\ge \tfrac{1}{2}
            \end{cases}
    \end{align}
    where $c\eqdef e$.
    Stacking this, we extend~\eqref{eq:basic_activation} to $\bar{\sigma}:\{0,1\}^2\to \{0,1\}^{16}$ sending any $\{0,1\}^2$ to
    \begin{equation}
    \label{eq:basic_activation__total}
            \bar{\sigma}
            (x)
        \eqdef
            \big(
                \tilde{\sigma}
                (x|Z)
            \big)_{Z\in \{0,1\}^4}
    .
    \end{equation}
    We now create an $\mathbb{R}^{16}$-valued non-linearity whose coordinates are simply $(\tilde{\sigma}(x|Z)$) with the $Z$ running through the different values of the four dimensional Boolean hypercube. Namely, we consider the non-linearity $\sigma:\mathbb{R}^2\to \mathbb{R}^{16}$ given for any $x\in \mathbb{R}^2$ by
    \begin{equation}
    \label{eq:EveryPossibleGateAllAtOnce}
            \quad \sigma(x) 
        =
            \big(
                \tilde{\sigma}
                (x|Z^{(i)})
            \big)_{i=1}^{16},
    \end{equation}
    where $Z^{(1)}=(0,0,0,0),\dots,Z^{(16)}=(1,1,1,1)$ is an enumeration of $\{0,1\}^4$.

Consider the stochastic Boolean gate selector defined for any trainable parameter
$W_{\Sigma}\in\mathbb{R}^{16}$ and any input $x\in\mathbb{R}^2$ by
$\Sigma(x\mid W_{\Sigma}) \eqdef \operatorname{SM}_{16}(W_{\Sigma})^{\top}\sigma(x)$.
The quantity is the pointwise mean, of any random variable $\randactiv^{W_\Sigma}:\Omega\to C^{\infty}(\mathbb{R}^2,\mathbb{R})$ given for any outcome $\omega \in \Omega$ by
\begin{equation}
\label{eq:Random_Gate}
\begin{aligned}
        \randactiv^{W_\Sigma}(\omega)(x)
    & \eqdef 
        \sum_{i=1}^{16}\,
            X_i^{W_\Sigma}(\omega)
            \, 
            \tilde{\sigma}(x|Z^{(i)})
\\
\mbox{ s.t. }
    \mathbb{P}(X^{W_\Sigma}=e_i) & \eqdef \big(\operatorname{SM}_{16}({W_\Sigma}) \big)_i
\end{aligned}
\end{equation}
for each $i\in [16]_+$, 
where $x$ is an arbitrary element of $\mathbb{R}^2$.

\subsubsection{Random Circuit Network Recursion}
\label{s:Model__ss:Model}
Fix a depth parameter $\Delta \in \mathbb{N}_+$, widths $(B_{in},\dots,B_{out})\in \mathbb{N}_+^{\Delta+1}$, and a temperature parameter $\eta\ge 0$.
For every layer  $l\in [\Delta-1]_+$ and every set of node weights $W^{(l:\sigma)}\in \mathbb{R}^{B_{l+1}\times 16}$, 
we define the \textbf{stochastic Boolean gate}/activation function $\randactiv^{W^{(l:\sigma)}}:\mathbb{R}^{B_{l+1}}\times 
\mathbb{R}^{B_{l+1}}
\to \mathbb{R}^{B_{l+1}}$ as mapping every  $(x,y)\in \mathbb{R}^{B_{l+1}}\times \mathbb{R}^{B_{l+1}}$ to
\begin{equation}
\label{eq:random_activation}
        \randactiv^{W^{(l:\sigma)}}(x, y)
    \eqdef 
        \big(
            \randactiv^{W^{(l:\sigma)}_j}
            (x_j,y_j)
        \big)_{j=1}^{B_{l+1}}
.
\end{equation}
Then, $\randactiv^{W^{(l:\sigma)}}$ specifies each gate at every node of the stochastic Boolean circuit implementing our model as parameterized, by the trainable matrix $W^{(l:\sigma)}$.  
At every layer $ l\in [\Delta - 1]_+$ the edges in our stochastic Boolean circuit are defined by the 
$B_{l+1}\times B_l$ 
\begin{equation}
\label{eq:edge_selector_mean}
        \randEdge^{(l:1)}
    \eqdef 
        \begin{bmatrix}
            \operatorname{SM}(\eta W^{(l:1)}_1)
        \\ 
        \vdots
        \\
            \operatorname{SM}(\eta W^{(l:1)}_{B_{l+1}}) 
        \end{bmatrix}
\mbox{ and }
        \randEdge^{(l:2)}
    \eqdef 
        \begin{bmatrix}
                \operatorname{SM}\Big(
                        \eta\big(1-\operatorname{SM}(
                            \eta W^{(l:1)}_1
                        )
                    \big)
                \odot
                    \operatorname{SM}(\eta\,W_1^{(l:2)})
            \Big)
        \\ 
        \vdots
        \\
            \operatorname{SM}\Big(
                        \eta\big(1-\operatorname{SM}(
                            \eta W^{(l:1)}_{B_{l+1}}
                        )
                    \big)
                \odot
                    \operatorname{SM}(\eta\,W_{B_{l+1}}^{(l:2)})
            \Big)
        \end{bmatrix}
,
\end{equation}
where $\odot$ denotes the componentwise (Hadamard) product on pairs of vectors in $\mathbb{R}^{B_l}$.
Let $\randEdge^{(l:1:e)}$ and $\randEdge^{(l:2:e)}$ denote the corresponding
sampled adjacency matrices whose row laws are specified by~\eqref{eq:edge_selector_mean}
(i.e., for each $j\in[B_{l+1}]_+$, the $j$-th row of $\mathbf{E}^{(l:1:e)}$ and
$\mathbf{E}^{(l:2:e)}$ is drawn as a one-hot vector in $\{e_1,\dots,e_{B_l}\}$
with mean given by the $j^{th}$ row of $E^{(l:1)}$ and $E^{(l:2)}$, respectively).
We may now formally define our \ModelNameShort\ model.
\begin{definition}[The \ModelName\ Model]
\label{defn:Model}
Fix $B_{in},B^{\uparrow},\Delta,B_{out}\in \mathbb{N}_+$.  The \ModelName\ is a map $\randcirc:\Omega\times \{0,1\}^{B_{in}}\to \{0,1\}^{B_{out}}$ defined by iteratively mapping any $x\in \{0,1\}^{B_{in}}$ to\footnote{Following standard probability theory convention, we suppress the dependence of the stochastic Boolean gates $\{\randactiv^{W^{(l:\sigma)}}\}_{l=1}^{\Delta-1}$, random output of the bit lifting channel  $\mathbf{X}^{\uparrow:W}$, and random adjacency matrices $\big\{\randEdge^{(l:j:e)}\big\}_{l=1,\; j=1}^{\Delta-1,\;2}$ on the outcome $\omega \in \Omega$.}
\begin{equation}
\label{eq:Model_Def}
\begin{aligned}
        \randcirc(x)
    & \eqdef 
        \boldsymbol{x}^{(\Delta)}
\\
        \boldsymbol{x}^{(l+1)}
    & \eqdef
        \randactiv^{W^{(l:\sigma)}}\big(
                \randEdge^{(l:1:e)}\,
                    \boldsymbol{x}^{(l)}
            ,
                \randEdge^{(l:2:e)}\,
                    \boldsymbol{x}^{(l)}
        \big)
    \mbox{ \qquad for } l=1,\dots,\Delta -1 
\\
        \boldsymbol{x}^{(1)}
    & \eqdef 
        \mathbf{X}^{\uparrow:W}(x) 
,
\end{aligned}
\end{equation}
where the random matrices $\randEdge^{(l:1:e)}$ and $\randEdge^{(l:2:e)}$ are as in~\eqref{eq:edge_selector_mean}, the stochastic Boolean gates $\randactiv^{W^{(l:\sigma)}}$ are as in~\eqref{eq:random_activation}, and the stochastic bit-lifting channel $\mathbf{X}^{\uparrow:W}$ is as in~\eqref{eq:lifted_guy}.  For $l=1,\dots,\Delta$, the dimensions of these objects are such that the composition defining $\randcirc$ is well-defined for every $(\omega,x)\in \Omega\times \{0,1\}^{B_{in}}$.
\end{definition}

\subsection{Explaining Each Component of the \ModelNameShort\ Model}
We now justify our parameterizations of each component of the proposed \ModelName\ model by briefly summarizing their key interpretability and expressivity properties.

For any choice of weight matrix $W$, the stochastic bit-lifting channel $\mathcal{L}(\cdot\,|\,W)$ always defines a random function taking values in the $B^{\uparrow}$-dimensional Boolean cube. Moreover, $W$ can be chosen, during training, so that any desired subset of variables in $\{0,1\}^B$, or their negations, is repeated or omitted an arbitrary number of times.  
\begin{proposition}[Stochastic Bit-Lifting Channels]
\label{prop:lifting_channels}
Let $B,B^{\uparrow}\in \mathbb{N}_+$, let $W\in \mathbb{R}^{B^{\uparrow}\times 2B}$, and let $\mathbf{X}^{\uparrow:W}$ be defined in~\eqref{eq:lifted_guy}.  
Then, $\mathbf{X}^{\uparrow:W}(x)\in \{0,1\}^{B^{\uparrow}}$
$\mathbb{P}$-a.s.
\hfill\\
Moreover, if $W$ is binary with only one non-zero entry per row then: for every $0<\delta \le 1$ there exists some $\eta>0$ such that 
\[
\resizebox{1\hsize}{!}{$
        \mathbb{P}\Biggl(
            \big(\forall i \in [B]_+\big)\,
                    \#\big\{
                        X_j^{\uparrow:
                        {\eta}
                        W}=x_i
                    \big\}_{j\in [B^{\uparrow}]_+}
                =   
                    \|W_{:i}\|_0
            \mbox{ and }
                    \#\big\{
                        X_j^{\uparrow:
                        {\eta}
                        W}=\neg\,x_i
                    \big\}_{j\in [B^{\uparrow}]_{+}}
                =   
            \|W_{:i+B}\|_0
        \Biggr)
    \ge
        1-\delta
$}
.
\]
\end{proposition}
Next, we show that each neuron in \ModelNameShort\ selects exactly one pair of incoming edges from the previous layer and computes exactly one of the $16$ fan-in $2$ fan-out $1$ Boolean gates.
\begin{proposition}[Only, and Any, Two Edges Can Be Sampled]
\label{prop:Good_Properties_Uncrossable}
Let $B\in \mathbb{N}_+$. For all $\eta > 0$,
\begin{itemize}
    \item[(i)] \textbf{Non-Singular Probabilistic Interpretability:} 
    The range of $\operatorname{S}_{\eta}$ is $\dot{\Delta}_B\times \dot{\Delta}_B$,
    \item[(ii)] \textbf{Smoothness:} The map $\operatorname{S}_{\eta}$ is ($C^{\infty}$-)smooth.
\end{itemize}
Furthermore,   
For every $0<\varepsilon<\tfrac{1}{2}$, there exists a temperature threshold $\eta_{\varepsilon}>0$ such that: for each $i\in [B]_+$,
\begin{itemize}
    \item[(iii)] \textbf{No Doubled Edges:}
    $
        \sup\limits_{\substack{
    W_2 \in \mathbb{R}^B \\
    \exists\, j \neq i: \hspace{0.1cm} (W_2)_j > (W_2)_i
}}\,
            \operatorname{S}_{\eta_{\varepsilon}}\big(
                e_i
                ,
                W_2
            \big)_{i:2}
        \in 
            [0,\varepsilon].
    $
    \item[(iv)] \textbf{Recovery}: If $j\in [B]_+$ and $j\neq i$ then $
        \big\|
                \operatorname{S}_{\eta_{\varepsilon}}\big(
                    e_i
                    ,
                    e_j
                \big)
            -
                (e_i,e_j)
        \big\|_1
        <
        \varepsilon
    $.
\end{itemize}
In particular, if $(\mathbf{E}_{i:1},\mathbf{E}_{i:2})\sim S_{\eta_{\varepsilon}}(e_i,e_j)$ then
$
        \mathbb{P}\big(
            (\mathbf{E}_{i:1},\mathbf{E}_{i:2})
            =
            (e_i,e_j)
        \big)
    \ge 
        1-\varepsilon
$.
\end{proposition}
We show that stochastic Boolean gates are meaningful; i.e., they are in $\mathbb{G}^\star$, and that we can set a stochastic Boolean gate to be a \emph{specific} gate in $\mathbb{G}^\star$ with arbitrarily high probability.
\begin{proposition}[Stochastic Boolean Gates: Sample Boolean Gates]
\label{prop:Valid_gate}
\hfill\\
For every $W\in \mathbb{R}^{16}$, $\randactiv^W|_{\{0,1\}^2}\in \Bary$.
\end{proposition}
\begin{proposition}[Stochastic Boolean Gates: Universally Sample Boolean Gates]
\label{prop:random_weight_assignment}
For every failure probability $\delta\in (0,1]$ and each $i\in [16]_+$, if $\eta\ge \ln\big( \frac{15}{\delta} - 15\big)$ then
\begin{equation}
\label{eq:prop:random_weight_assignment__TVBound}
        \operatorname{TV}\big(
            \operatorname{Law}(\randactiv^{\eta e_i}|_{\{0,1\}^2})
            ,
                \delta_{g_i}
        \big)
    \le 
        \delta
.
\end{equation}
In particular, $
        \mathbb{P}\big(
            \randactiv^{\eta e_i}|_{\{0,1\}^2}  = g_i
        \big)
    \ge 
        1-\delta
$.
\end{proposition}

\section{Main Results}
\label{s:Theory}

We are now prepared to present our first main result which \textit{certifies} that any possible configuration of the parameters of \ModelName \ must yield a distribution over Boolean circuits with fan-in $2$ and fan-out $1$.  Consequently, a sample from such a configuration must a.s.\ be such a Boolean circuit of the form~(i); page 2.
\begin{theorem}[Certifiable Boolean Reasoning]
\label{thrm:StructureAware}
Fix a depth $\Delta\in\mathbb{N}_+$, a width $\Upsilon\in \mathbb{N}_+$, 
bits $B_{in}=B_0,\dots,B_{\Delta}=B_{out}=1\in \mathbb{N}_+$.  Then, for any matrices $W$, $\{W^{(l:\sigma)}\}_{l=1}^{\Delta-1}$  and $\{\mathbf{E}^{(l:1:e)}\}_{l=1}^{\Delta-1}$, $\{\mathbf{E}^{(l:2:e)}\}_{l=1}^{\Delta-1}$ as in~\eqref{eq:Model_Def} the (random) function  $\randcirc$
is such that 
\[
    \mathbb{P}\big(
     \randcirc
    \in \operatorname{CIRCUIT}_{\Delta,\Upsilon}^{B_{in}}\big)
    =
    1
.
\]
\end{theorem}
Having guaranteed that the \ModelName \ a.s. samples Boolean circuits of the form (i), we now show that it is \textit{universal} in the sense that, provided the model uses sufficiently many neurons, its parameters can be configured so as to sample any prescribed Boolean circuit with arbitrarily high probability.
\begin{theorem}[Universal Boolean Reasoning]
\label{thrm:UniversalReasoning}
Let $B_{in}\in \mathbb{N}_+$ and let $f:\{0,1\}^{B_{in}}\to \{0,1\}$.  For every failure probability $0<\delta \le 1$, there exists an $
 \randcirc :\Omega \times \{0,1\}^{B}\to \{0,1\}$ with representation~\eqref{eq:Model_Def} such that
\begin{equation}
\label{eq:universality}
    \mathbb{P}\Big(
        \randcirc(x) = f(x)
        \:\,
        \forall x\in \{0,1\}^{B_{in}}
        \,
    \Big)
\ge 
    1-\delta
.
\end{equation}
\end{theorem}
Theorem~\ref{thrm:UniversalReasoning} is a direct analogue of classical universal approximation theorems for MLPs for continuous function approximation, e.g.~\cite{funahashi1989approximate,cybenko1989approximation,hornik1989multilayer,pinkus1999approximation,kidger2020universal}, in that it provides an existence guarantee independent of the “hardness” of representing (resp.\ approximating) the ground-truth function. Naturally, one wonders whether there are instances in which favourable rates are achievable; which we now confirm is possible.

\subsection{Efficient Reasoning}

Indeed, analogously to quantitative universal approximation results for classical deep learning models, explicit approximation rates, while optimal in the worst-case cf.~\cite{yarotsky2018optimal}, are often not the most compelling. Rather, the central question is whether the number of parameters can scale favourably with the intrinsic degrees of freedom of the problem. In classical neural network approximation theory, such favourable scaling is typically obtained under structural assumptions on the target function, notably some form of smoothness~\cite{yarotsky2017error,suzuki2018adaptivity,lu2021deep,guhring2021approximation,JMLR:v24:23-0025}. In our Boolean setting, an analogous notion of ``nice structure'' is sparsity, formalized by the target function being an $S$-junta with $S\ll B$; i.e.\ the target function only depends on $S$ input bits out of the total $B$.  We now show that, when the sparsity parameter $S$ scales logarithmically with the number of inputs $B$, the number of neurons required by the \ModelName~in Theorem~\ref{thrm:UniversalReasoning} is polynomial; and in cases even linear, up to logarithmic factors.

\begin{theorem}[Efficient Boolean Reasoning for $\mathcal{O}(\log(B))$-Juntas]
\label{thrm:UniversalReasoning__Quantitative}
Let $B,S\in \mathbb{N}_+$, $C>0$, with $S\le C\log(B)$ and let $f:\{0,1\}^B\to \{0,1\}$ be an $S$-junta. 
For every 
$0<\delta \le 1$, there is an $
\randcirc :\Omega \times \{0,1\}^{B}\to \{0,1\}$ of depth $\Delta$, width $\Upsilon$, and with at-most $N$ neurons satisfying~\eqref{eq:universality}; where
\[
        \Delta
    \in 
        \mathcal{O}\big(
            \log(B) + \log(\log(B))
        \big)
,\,\,
    \Upsilon \in \mathcal{O}(B)
,\mbox{ and}\,\,
        N,B^{\uparrow}
    \in 
        \mathcal{O}\big(
            B^p
            \log(B)
        \big)
\]
where $p=C>0$.
In particular, if $0<C\le 1$, then $
N,B^{\uparrow}
    \in 
        \mathcal{O}\big(
            B
            \log(B)
        \big)
$.
\end{theorem}
Having verified our theoretical desiderata (i)--(iii), we now confirm the experimental performance (iv) of our \ModelName~model.
\section{Experimental Validation}
\label{s:Experiments}

We empirically validate two claims: (i) the proposed model can be trained to achieve competitive accuracy on propositional reasoning tasks, and (ii) its internal computation admits a faithful circuit-level interpretation.

\paragraph{Setup.}
Each instance is a Boolean function $f:\{0,1\}^B\to\{0,1\}$ specified by a syntactic formula. We train on the full truth table (all $2^B$ assignments) and evaluate \emph{exact match} (EM), i.e., the predicted truth table must agree with the target on all $2^B$ inputs. We compare against ReLU MLP baselines under three matching regimes controlling architectural or parameter budgets (neuron-match and two parameter-matching variants). Dataset generation, matching rules, and training details are provided in Appendix~\ref{app:exp_impl}.

\paragraph{Predictive performance.}
Table~\ref{tab:em_main} reports EM across five seeds. Matched MLPs achieve slightly higher EM in this regime; nonetheless, our model remains competitive while enforcing a circuit-structured computation. Since EM is computed on the full truth table, the reported differences reflect exact correctness rather than partial-credit metrics.

\begin{table}[t]
\centering
\small
\setlength{\tabcolsep}{6pt}
\renewcommand{\arraystretch}{1.15}
\begin{tabular}{l c c c}
\toprule
MLP match & Avg.\ MLP params & MLP EM & \ModelNameShort\ EM \\
\midrule
neuron & 368 & 0.987 $\pm$ 0.002 & \multirow{3}{*}{0.971 $\pm$ 0.003} \\
param\_soft & 683 & 0.992 $\pm$ 0.001 & \\
param\_total & 1154 & 0.999 $\pm$ 0.002 & \\
\bottomrule
\end{tabular}
\caption{Exact-match (EM) on truth-table evaluation (mean $\pm$ std across 5 seeds). \ModelNameShort\ EM is reported once under the canonical \ModelNameShort\ setting; MLP EM is shown for each matching regime.}
\label{tab:em_main}
\end{table}

\paragraph{Neuronwise Booleanity (BNR).}
\hspace{-0.25cm} We measure \emph{Boolean Neuronal Representability} (BNR), introduced in
Definition~\ref{defn:BnB} and characterized in Proposition~\ref{prop:BNR_Equivalence}.
Intuitively, a unit is BNR if it takes at most two distinct values on the Boolean cube and can therefore be
recoded injectively into $\{0,1\}$.
For continuous baselines (ReLU MLPs), we compute practical diagnostics by enumerating each unit over the full
truth table and applying (i) a rounded two-valued check (BNR$_{\mathrm{exact}}$) and (ii) a tolerance-based
two-cluster check (BNR$_\varepsilon$); see Appendix~\ref{app:interpretability}.

\paragraph{Certifiable interpretability gap.}
Table~\ref{tab:bnr_main} shows a pronounced separation: our model attains BNR$=1$ for all internal
units, whereas matched MLP baselines exhibit low neuronwise Booleanity (roughly $10\%$ at the first
hidden layer and $\approx 20\%$ on average across hidden layers).
This captures the intended certificate-level distinction: our model implements neuronwise Boolean
computation by design, while standard MLP activations typically are not two-valued on the Boolean cube.
Full experimental protocol details (including matching regimes, decoding, and evaluation) are provided
in Appendix~\ref{app:exp_impl}.

\begin{table}[t]
\centering
\small
\setlength{\tabcolsep}{6pt}
\renewcommand{\arraystretch}{1.15}
\begin{tabular}{lcccc}
\toprule
Model / match
& BNR$_{\mathrm{exact}}$(L1)
& BNR$_{\mathrm{exact}}$(all)
& BNR$_{\varepsilon}$(L1)
& BNR$_{\varepsilon}$(all)
\\
\midrule
\ModelNameShort\ & 1.000 & 1.000 & 1.000 & 1.000 \\
\midrule
neuron & 0.114 $\pm$ 0.003 & 0.209 $\pm$ 0.002 & 0.115 $\pm$ 0.004 & 0.209 $\pm$ 0.002 \\
param\_soft & 0.107 $\pm$ 0.002 & 0.199 $\pm$ 0.001 & 0.107 $\pm$ 0.002 & 0.199 $\pm$ 0.001 \\
param\_total & 0.105 $\pm$ 0.002 & 0.195 $\pm$ 0.002 & 0.105 $\pm$ 0.002 & 0.196 $\pm$ 0.002 \\
\bottomrule
\end{tabular}
\caption{BNR diagnostics (mean $\pm$ std across 5 seeds). ``all'' denotes the average across hidden layers.
Formal definitions and numerical procedures are detailed in Appendix~\ref{app:BNR_detail} and Appendix~\ref{app:interpretability}.}
\label{tab:bnr_main}
\end{table}

\noindent Additional gate-recoverability probes and gate-type histograms are reported in
Appendix~\ref{app:prim_table} and Appendix~\ref{app:gate_fig}; they provide complementary context for
interpreting BNR and EM.

\section{Conclusion}
\label{s:Conclusion}

\noindent
We proposed the \ModelName\ (\ModelNameShort), a trainable distribution over fan-in $2$, fan-out $1$ Boolean circuits.  
The model provides \emph{certifiable reasoning} for all parameters (Theorem~\ref{thrm:StructureAware}) and \emph{universal Boolean reasoning} via suitable parameter choices (Theorem~\ref{thrm:UniversalReasoning}).  
In sparse regimes, when the target function is an $\mathcal{O}(\log B)$-junta, this universality can be achieved with essentially linear scaling in the ambient dimension (Theorem~\ref{thrm:UniversalReasoning__Quantitative}).  

\noindent
Experiments validate the \ModelName\ architecture, showing that it represents diverse Boolean functions while almost surely sampling fan-in $2$, fan-out $1$ Boolean circuits. The model attains perfect neuron-wise Booleanity (BNR$=1$), whereas matched ReLU MLPs achieve similar accuracy but exhibit low BNR, raising concerns about circuit-level reasoning in classical architectures.

\subsection*{Future work} 
Future work will explore trainable distributions over richer circuit and program classes with comparable structural guarantees.


\acks{
A.\ Kratsios, H.\ Ghoukasian, and D.\ Zvigelsky acknowledge the financial support from the NSERC Discovery Grant No.\ RGPIN-2023-04482 and No.\ DGECR-2023-00230.  A.\ Kratsios was also supported by the project Bando PRIN 2022 named ``Qnt4Green - Quantitative Approaches for Green Bond Market: Risk Assessment, Agency Problems and Policy Incentives'', codice 2022JRY7EF, CUP E53D23006330006, funded by European Union – NextGenerationEU, M4c2.
}

\bibliography{Bookkeaping/0_refs}

@inproceedings{cobbe2021training,
  title={Training Verifiers to Solve Math Word Problems},
  author={Cobbe, Karl and Kosaraju, Vineet and Bavarian, Mohammad and Chen, Mark and Jun, Heewoo and Kaiser, Lukasz and Plappert, Matthias and Tworek, Jerry and Hilton, Jacob and Nakano, Reiichiro and others},
  booktitle={arXiv preprint arXiv:2110.14168},
  year={2021}
}

@inproceedings{hendrycks2021measuring,
  title={Measuring Mathematical Problem Solving With the MATH Dataset},
  author={Hendrycks, Dan and Burns, Collin and Kadavath, Saurav and Arora, Akul and Basart, Steven and Tang, Eric and Song, Dawn and Steinhardt, Jacob},
  booktitle={arXiv preprint arXiv:2103.03874},
  year={2021}
}

@article{yarotsky2017error,
  title={Error bounds for approximations with deep ReLU networks},
  author={Yarotsky, Dmitry},
  journal={Neural Networks},
  volume={94},
  pages={103--114},
  year={2017},
  publisher={Elsevier}
}

@article{lu2021deep,
  title={Deep network approximation for smooth functions},
  author={Lu, Jianfeng and Shen, Zuowei and Yang, Haizhao and Zhang, Shijun},
  journal={SIAM Journal on Mathematical Analysis},
  volume={53},
  number={5},
  pages={5465--5506},
  year={2021},
  publisher={SIAM}
}

@article{guhring2021approximation,
  title={Approximation rates for neural networks with encodable weights in smoothness spaces},
  author={G{\"u}hring, Ingo and Raslan, Mones},
  journal={Neural Networks},
  volume={134},
  pages={107--130},
  year={2021},
  publisher={Elsevier}
}

@inproceedings{suzuki2018adaptivity,
title={Adaptivity of deep Re{LU} network for learning in Besov and mixed smooth Besov spaces: optimal rate and curse of dimensionality},
author={Taiji Suzuki},
booktitle={International Conference on Learning Representations},
year={2019},
url={https://openreview.net/forum?id=H1ebTsActm},
}

@article{JMLR:v24:23-0025,
  author  = {Jonathan W. Siegel},
  title   = {Optimal Approximation Rates for Deep ReLU Neural Networks on Sobolev and Besov Spaces},
  journal = {Journal of Machine Learning Research},
  year    = {2023},
  volume  = {24},
  number  = {357},
  pages   = {1--52},
  url     = {http://jmlr.org/papers/v24/23-0025.html}
}

@inproceedings{yarotsky2018optimal,
  title={Optimal approximation of continuous functions by very deep ReLU networks},
  author={Yarotsky, Dmitry},
  booktitle={Conference on learning theory},
  pages={639--649},
  year={2018},
  organization={PMLR}
}

@book{pinkus1999approximation,
  title     = {Approximation Theory of the {MLP} Model in Neural Networks},
  author    = {Pinkus, Allan},
  publisher = {Cambridge University Press},
  year      = {1999},
  series    = {Cambridge Monographs on Applied and Computational Mathematics},
}

@inproceedings{kidger2020universal,
  title={Universal approximation with deep narrow networks},
  author={Kidger, Patrick and Lyons, Terry},
  booktitle={Conference on learning theory},
  pages={2306--2327},
  year={2020},
  organization={PMLR}
}

@article{hornik1989multilayer,
  title={Multilayer feedforward networks are universal approximators},
  author={Hornik, Kurt and Stinchcombe, Maxwell and White, Halbert},
  journal={Neural networks},
  volume={2},
  number={5},
  pages={359--366},
  year={1989},
  publisher={Elsevier}
}

@article{cybenko1989approximation,
  title={Approximation by superpositions of a sigmoidal function},
  author={Cybenko, George},
  journal={Mathematics of control, signals and systems},
  volume={2},
  number={4},
  pages={303--314},
  year={1989},
  publisher={Springer}
}

@article{funahashi1989approximate,
  title   = {On the approximate realization of continuous mappings by neural networks},
  author  = {Funahashi, Kenji},
  journal = {Neural Networks},
  volume  = {2},
  number  = {3},
  pages   = {183--192},
  year    = {1989},
  doi     = {10.1016/0893-6080(89)90003-8}
}

@article{mcculloch1943logical,
  title={A logical calculus of the ideas immanent in nervous activity},
  author={McCulloch, Warren S. and Pitts, Walter},
  journal={The Bulletin of Mathematical Biophysics},
  volume={5},
  number={4},
  pages={115--133},
  year={1943}
}

@inproceedings{dai-etal-2022-knowledge,
  title = {Knowledge Neurons in Pretrained Transformers},
  author = {Dai, Damai and Dong, Li and Hao, Yaru and Sui, Zhifang and Chang, Baobao and Wei, Furu},
  booktitle = {Proceedings of the 60th Annual Meeting of the Association for Computational Linguistics (Volume 1: Long Papers)},
  year = {2022},
  pages = {8493--8502},
  address = {Dublin, Ireland},
  publisher = {Association for Computational Linguistics}
}

@article{Alain2016linearprobes,
  title = {Understanding intermediate layers using linear classifier probes},
  author = {Alain, Guillaume and Bengio, Yoshua},
  journal = {arXiv preprint arXiv:1610.01644},
  year = {2016}
}

@book{cook_nguyen_2010_logical_foundations,
  title     = {Logical Foundations of Proof Complexity},
  author    = {Cook, Stephen A. and Nguyen, Phuong},
  year      = {2010},
  publisher = {Cambridge University Press},
  series    = {Perspectives in Logic},
  address   = {Cambridge},
  isbn      = {9780521517225}
}

@misc{chaudhari2026exploiting,
    title={Exploiting Reasoning Patterns in Language Models for Indirect Targeted Poisoning},
    author={Harsh Chaudhari and Ethan Rathbun and Hanna Foerster and Jamie Hayes and Matthew Jagielski and Milad Nasr and Ilia Shumailov and Alina Oprea},
    year={2026},
    url={https://openreview.net/forum?id=hSLopCTOtT}
}

@inproceedings{zou2025poisonedrag,
  title={$\{$PoisonedRAG$\}$: Knowledge corruption attacks to $\{$Retrieval-Augmented$\}$ generation of large language models},
  author={Zou, Wei and Geng, Runpeng and Wang, Binghui and Jia, Jinyuan},
  booktitle={34th USENIX Security Symposium (USENIX Security 25)},
  pages={3827--3844},
  year={2025}
}

@inproceedings{liu2024formalizing,
  title={Formalizing and benchmarking prompt injection attacks and defenses},
  author={Liu, Yupei and Jia, Yuqi and Geng, Runpeng and Jia, Jinyuan and Gong, Neil Zhenqiang},
  booktitle={33rd USENIX Security Symposium (USENIX Security 24)},
  pages={1831--1847},
  year={2024}
}

@inproceedings{rajeev2025cats,
title={Cats Confuse Reasoning {LLM}: Query Agnostic Adversarial Triggers for Reasoning Models},
author={Meghana Arakkal Rajeev and Rajkumar Ramamurthy and Prapti Trivedi and Vikas Yadav and Oluwanifemi Bamgbose and Sathwik Tejaswi Madhusudhan and James Zou and Nazneen Rajani},
booktitle={Second Conference on Language Modeling},
year={2025},
url={https://openreview.net/forum?id=VrEPiN5WhM}
}

@article{song2025thinking,
  title={Thinking Isn't an Illusion: Overcoming the Limitations of Reasoning Models via Tool Augmentations},
  author={Song, Zhao and Yue, Song and Zhang, Jiahao},
  journal={arXiv preprint arXiv:2507.17699},
  year={2025}
}

@article{olsson2022induction,
  title   = {In-Context Learning and Induction Heads},
  author  = {Olsson, Catherine and Elhage, Nelson and Nanda, Neel and Joseph, Nicholas and DasSarma, Nova and Henighan, Tom and Mann, Ben and Askell, Amanda and Bai, Yuntao and Chen, Anna and others},
  journal = {arXiv preprint arXiv:2209.11895},
  year    = {2022},
  url     = {https://arxiv.org/abs/2209.11895}
}

@inproceedings{shojaee2025the,
title={The Illusion of Thinking: Understanding the Strengths and Limitations of Reasoning Models via the Lens of Problem Complexity},
author={Parshin Shojaee and Seyed Iman Mirzadeh and Keivan Alizadeh and Maxwell Horton and Samy Bengio and Mehrdad Farajtabar},
booktitle={The Thirty-ninth Annual Conference on Neural Information Processing Systems},
year={2025},
url={https://openreview.net/forum?id=YghiOusmvw}
}

@inproceedings{wu2024reasoning,
  title     = {Reasoning or Reciting? Exploring the Capabilities and Limitations of Language Models through Counterfactual Tasks},
  author    = {Wu, Zhaofeng and Qiu, Linlu and Ross, Alexis and Aky{\"u}rek, Ekin and Chen, Boyuan and Wang, Bailin and Kim, Najoung and Andreas, Jacob and Kim, Yoon},
  booktitle = {Proceedings of the Association for Computational Linguistics (ACL)},
  year      = {2024},
  publisher = {Association for Computational Linguistics}
}

@inproceedings{keysers2020measuring,
  title     = {Measuring Compositional Generalization: A Comprehensive Method on Realistic Data},
  author    = {Keysers, Daniel and Sch{\"a}rli, Nathanael and Scales, Nathan and Buisman, Hylke and Furrer, Daniel and Kashubin, Sergey and Momchev, Nikola and Sinopalnikov, Danila and Stafiniak, Lukasz and Tihon, Tomas and others},
  booktitle = {International Conference on Learning Representations (ICLR)},
  year      = {2020},
  url       = {https://arxiv.org/abs/1912.09713}
}

@inproceedings{bender2021stochastic,
  title     = {On the Dangers of Stochastic Parrots: Can Language Models Be Too Big?},
  author    = {Bender, Emily M. and Gebru, Timnit and McMillan-Major, Angelina and Shmitchell, Shmargaret},
  booktitle = {Proceedings of the ACM Conference on Fairness, Accountability, and Transparency (FAccT)},
  year      = {2021},
  pages     = {610--623},
  url       = {https://arxiv.org/abs/2101.11770}
}

@inproceedings{wei2022chain,
  title     = {Chain-of-Thought Prompting Elicits Reasoning in Large Language Models},
  author    = {Wei, Jason and Wang, Xuezhi and Schuurmans, Dale and Bosma, Maarten and Ichter, Brian and Xia, Fei and Chi, Ed and Le, Quoc and Zhou, Denny},
  booktitle = {Advances in Neural Information Processing Systems (NeurIPS)},
  year      = {2022},
  url       = {https://arxiv.org/abs/2201.11903}
}

@article{berglund2023reversal,
  title   = {The Reversal Curse: LLMs Trained on ``A is B'' Fail to Learn ``B is A''},
  author  = {Berglund, Leo and et al.},
  journal = {arXiv preprint arXiv:2309.12288},
  year    = {2023},
  url     = {https://arxiv.org/abs/2309.12288}
}

@inproceedings{zhao2024mathtrap,
  title     = {MathTrap: Evaluating Mathematical Reasoning of Large Language Models via Trap Problems},
  author    = {Zhao, Jun and Li, Yifan and Zhang, Zhen and Chen, Yiming},
  booktitle = {Proceedings of the Conference on Empirical Methods in Natural Language Processing (EMNLP)},
  year      = {2024},
  url       = {https://aclanthology.org/2024.emnlp-main.915}
}

@article{song2024survey,
  title   = {A Survey on Large Language Model Reasoning Failures},
  author  = {Song, Peiyang and Han, Pengrui and Goodman, Noah D.},
  journal = {arXiv preprint arXiv:2410.08944},
  year    = {2024},
  url     = {https://arxiv.org/abs/2410.08944}
}

@article{mali2021neural,
  title={A neural state pushdown automata},
  author={Mali, Ankur Arjun and Ororbia II, Alexander G and Giles, C Lee},
  journal={IEEE Transactions on Artificial Intelligence},
  volume={1},
  number={3},
  pages={193--205},
  year={2021},
  publisher={IEEE}
}

@inproceedings{lowd2008learning,
  title={Learning arithmetic circuits},
  author={Lowd, Daniel and Domingos, Pedro},
  booktitle={Proceedings of the Twenty-Fourth Conference on Uncertainty in Artificial Intelligence},
  pages={383--392},
  year={2008}
}

@inproceedings{bournez2025universal,
  title     = {A Universal Uniform Approximation Theorem for Neural Networks},
  author    = {Olivier Bournez and Johanne Cohen and Adrian Wurm},
  booktitle = {Proceedings of the International Symposium on Mathematical Foundations of Computer Science (MFCS)},
  year      = {2025},
  address   = {Palaiseau, France},
  institution = {Institut Polytechnique de Paris, Université Paris-Saclay, BTU Cottbus-Senftenberg}
}

@article{karpinski1997polynomial,
  title={Polynomial bounds for VC dimension of sigmoidal and general Pfaffian neural networks},
  author={Karpinski, Marek and Macintyre, Angus},
  journal={Journal of Computer and System Sciences},
  volume={54},
  number={1},
  pages={169--176},
  year={1997},
  publisher={Elsevier}
}

@article{margulies2016polynomial,
  title={Polynomial-time solvable\# CSP problems via algebraic models and Pfaffian circuits},
  author={Margulies, Susan and Morton, Jason},
  journal={Journal of Symbolic Computation},
  volume={74},
  pages={152--180},
  year={2016},
  publisher={Elsevier}
}

@article{blumer1989learnability,
  title={\protect{Learnability and the Vapnik-Chervonenkis dimension}},
  author={Blumer, Anselm and Ehrenfeucht, Andrzej and Haussler, David and Warmuth, Manfred K},
  journal={Journal of the ACM (JACM)},
  volume={36},
  number={4},
  pages={929--965},
  year={1989},
  publisher={ACM New York, NY, USA}
}

@inproceedings{maass1991sigmoidvsbool,
  author    = {Wolfgang Maass and Georg Schnitger and Eduardo D. Sontag},
  title     = {On the Computational Power of Sigmoid versus Boolean Threshold Circuits},
  booktitle = {Proceedings of the 32nd Annual Symposium on Foundations of Computer Science (FOCS)},
  year      = {1991},
  pages     = {767--776}
}

@article{siu1993division,
  author  = {Kai-Yeung Siu and Jehoshua Bruck and Thomas Kailath and Thomas Hofmeister},
  title   = {Depth Efficient Neural Networks for Division and Related Problems},
  journal = {IEEE Transactions on Information Theory},
  volume  = {39},
  number  = {3},
  pages   = {946--956},
  year    = {1993},
  doi     = {10.1109/18.256501}
}

@article{merrill2022saturated,
  author  = {William Merrill and Ashish Sabharwal and Noah A. Smith},
  title   = {Saturated Transformers are Constant-Depth Threshold Circuits},
  journal = {Transactions of the Association for Computational Linguistics},
  year    = {2022}
}

@article{merrill2023parallelism,
  author  = {William Merrill and Ashish Sabharwal},
  title   = {The Parallelism Tradeoff: Limitations of Log-Precision Transformers},
  journal = {Transactions of the Association for Computational Linguistics},
  year    = {2023}
}

@article{chiang2024uniformtc0,
  author  = {Chiang, David},
  title   = {Transformers in Uniform {TC}$^0$},
  journal = {arXiv preprint},
  volume  = {arXiv:2409.13629},
  year    = {2024}
}

@book{siegelmann2012neural,
  title={Neural networks and analog computation: beyond the Turing limit},
  author={Siegelmann, Hava T},
  year={2012},
  publisher={Springer Science \& Business Media}
}

@inproceedings{perez2018on,
    title={On the Turing Completeness of Modern Neural Network Architectures},
    author={Jorge Pérez and Javier Marinković and Pablo Barceló},
    booktitle={International Conference on Learning Representations},
    year={2019},
    url={https://openreview.net/forum?id=HyGBdo0qFm},
}

@article{chung2021turing,
  title={Turing completeness of bounded-precision recurrent neural networks},
  author={Chung, Stephen and Siegelmann, Hava},
  journal={Advances in neural information processing systems},
  volume={34},
  pages={28431--28441},
  year={2021}
}

@inproceedings{chen2025magicore,
  title     = {MAGICORE: Multi-Agent, Iterative, Coarse-to-Fine Refinement for Reasoning},
  author    = {Chen, Justin Chih-Yao and Prasad, Archiki and Saha, Swarnadeep and Stengel-Eskin, Elias and Bansal, Mohit},
  booktitle = {Proceedings of EMNLP},
  year      = {2025},
  url       = {https://aclanthology.org/2025.emnlp-main.1660/}
}

@inproceedings{wang2023selfconsistency,
  title     = {Self-Consistency Improves Chain of Thought Reasoning in Language Models},
  author    = {Wang, Xuezhi and Wei, Jason and Schuurmans, Dale and Le, Quoc V. and Chi, Ed H. and Narang, Sharan and Chowdhery, Aakanksha and Zhou, Denny},
  booktitle = {International Conference on Learning Representations (ICLR)},
  year      = {2023},
  url       = {https://openreview.net/forum?id=1PL1NIMMrw}
}

@inproceedings{yao2023react,
  title     = {ReAct: Synergizing Reasoning and Acting in Language Models},
  author    = {Yao, Shunyu and Zhao, Jeffrey and Yu, Dian and Du, Nan and Shafran, Izhak and Narasimhan, Karthik R. and Cao, Yuan},
  booktitle = {International Conference on Learning Representations (ICLR)},
  year      = {2023},
  url       = {https://arxiv.org/abs/2210.03629}
}

@inproceedings{yao2023treeofthoughts,
  title     = {Tree of Thoughts: Deliberate Problem Solving with Large Language Models},
  author    = {Yao, Shunyu and Yu, Dian and Zhao, Jeffrey and Shafran, Izhak and Narasimhan, Karthik R. and Cao, Yuan},
  booktitle = {Advances in Neural Information Processing Systems (NeurIPS)},
  year      = {2023},
  url       = {https://arxiv.org/abs/2305.10601}
}

@inproceedings{schick2023toolformer,
  title     = {Toolformer: Language Models Can Teach Themselves to Use Tools},
  author    = {Schick, Timo and Dwivedi-Yu, Jane and Dess{\`i}, Roberto and Raileanu, Roberta and Lomeli, Maria and Zettlemoyer, Luke and Cancedda, Nicola and Scialom, Thomas},
  booktitle = {Advances in Neural Information Processing Systems (NeurIPS)},
  year      = {2023},
  url       = {https://arxiv.org/abs/2302.04761}
}

@inproceedings{shen2023hugginggpt,
  title     = {HuggingGPT: Solving AI Tasks with ChatGPT and Its Friends in HuggingFace},
  author    = {Shen, Yongliang and Chen, Lianmin and Tang, Xinyu and Li, Zhilin and Wang, Hui and Qian, Rui and Li, Chenliang and Jiang, Wenbo and Xiong, Yujie and Zhang, Wei and Liu, Zhuang and Liu, Peng},
  booktitle = {Advances in Neural Information Processing Systems (NeurIPS)},
  year      = {2023},
  url       = {https://arxiv.org/abs/2303.17580}
}

@inproceedings{shinn2023reflexion,
  title     = {Reflexion: Language Agents with Verbal Reinforcement Learning},
  author    = {Shinn, Noah and Labash, Beck and Gopinath, Ashwin and Narasimhan, Karthik R.},
  booktitle = {Advances in Neural Information Processing Systems (NeurIPS)},
  year      = {2023},
  url       = {https://arxiv.org/abs/2303.11366}
}

@article{madaan2023self,
  title={Self-refine: Iterative refinement with self-feedback},
  author={Madaan, Aman and Tandon, Niket and Gupta, Prakhar and Hallinan, Skyler and Gao, Luyu and Wiegreffe, Sarah and Alon, Uri and Dziri, Nouha and Prabhumoye, Shrimai and Yang, Yiming and others},
  journal={Advances in Neural Information Processing Systems},
  volume={36},
  pages={46534--46594},
  year={2023}
}

@inproceedings{wang2024voyager,
  title     = {Voyager: An Open-Ended Embodied Agent with Large Language Models},
  author    = {Wang, Guanzhi and Wang, Yu and Liu, Xinyi and Liu, Shibin and Chen, Yizhou and Zhang, Yiming and Zhao, Jian and Zhang, Tong},
  booktitle = {International Conference on Learning Representations (ICLR)},
  year      = {2024},
  url       = {https://arxiv.org/abs/2305.16291}
}

@inproceedings{zelikman2022star,
  title={Star: Self-taught reasoner},
  author={Zelikman, Eric and Wu, Yuhuai and Goodman, Noah D},
  booktitle={Proceedings of the NIPS},
  volume={22},
  year={2022}
}

@inproceedings{Hastad86,
  author    = {Johan H{\aa}stad},
  title     = {Almost Optimal Lower Bounds for Small Depth Circuits},
  booktitle = {Proceedings of the 18th Annual ACM Symposium on Theory of Computing (STOC)},
  year      = {1986},
  pages     = {6--20},
  publisher = {ACM}
}

@book{jukna2012boolean,
  title     = {Boolean Function Complexity: Advances and Frontiers},
  author    = {Jukna, Stasys},
  year      = {2012},
  publisher = {Springer},
  series    = {Algorithms and Combinatorics},
  volume    = {27},
  address   = {Berlin, Heidelberg},
  isbn      = {978-3-642-24507-5}
}

@article{kratsios2025quantifying,
  title={Quantifying The Limits of AI Reasoning: Systematic Neural Network Representations of Algorithms},
  author={Kratsios, Anastasis and Zvigelsky, Dennis and Hart, Bradd},
  journal={arXiv preprint arXiv:2508.18526},
  year={2025}
}

@book{bang2008digraphs,
  title={Digraphs: theory, algorithms and applications},
  author={Bang-Jensen, J{\o}rgen and Gutin, Gregory Z},
  year={2008},
  publisher={Springer Science \& Business Media}
}

@inproceedings{craven1996trepan,
 author = {Craven, Mark and Shavlik, Jude},
 booktitle = {Advances in Neural Information Processing Systems},
 editor = {D. Touretzky and M.C. Mozer and M. Hasselmo},
 pages = {},
 publisher = {MIT Press},
 title = {Extracting Tree-Structured Representations of Trained Networks},
 url = {https://proceedings.neurips.cc/paper_files/paper/1995/file/45f31d16b1058d586fc3be7207b58053-Paper.pdf},
 volume = {8},
 year = {1995}
}

@article{Badreddine2022LTN,
  title={Logic tensor networks},
  author={Badreddine, Samy and Garcez, Artur d'Avila and Serafini, Luciano and Spranger, Michael},
  journal={Artificial Intelligence},
  volume={303},
  pages={103649},
  year={2022},
  publisher={Elsevier}
}

@article{manhaeve2018deepproblog,
  title={Deepproblog: Neural probabilistic logic programming},
  author={Manhaeve, Robin and Dumancic, Sebastijan and Kimmig, Angelika and Demeester, Thomas and De Raedt, Luc},
  journal={Advances in neural information processing systems},
  volume={31},
  year={2018}
}

@inproceedings{PoonDomingos2011SPN,
  title={Sum-product networks: A new deep architecture},
  author={Poon, Hoifung and Domingos, Pedro},
  booktitle={2011 IEEE International Conference on Computer Vision Workshops (ICCV Workshops)},
  pages={689--690},
  year={2011},
  organization={IEEE}
}

@article{ProbCirc20,
  title={Probabilistic circuits: A unifying framework for tractable probabilistic models},
  author={Choi, Y and Vergari, Antonio and Van den Broeck, Guy},
  journal={UCLA. URL: http://starai. cs. ucla. edu/papers/ProbCirc20. pdf},
  pages={6},
  year={2020}
}
\newpage

\appendix

\section{Table of all 16 Boolean Gates with Fan-in 2 and Fan-out 1}
\label{s:GateEnumeration}
Table~\ref{tab:all_possible_gates} lists all possible Boolean gates with fan-in two and fan-out one.
\begin{table}[htp!]
    \centering
    \begin{tabular}{c l l| c c c c}
    & &  & \multicolumn{4}{c}{$Z^{(i)}$} \\ \hline
    \midrule
    Gate number & Gate & Expression & (0,0) & (0,1) & (1,0) & (1,1) \\ 
    \midrule
    $g_1$ & Constant $\operatorname{FALSE}$ & $0$ & 0 & 0 & 0 & 0 \\
    $g_2$ & AND & $A \land B$ & 0 & 0 & 0 & 1 \\
    $g_3$ & A AND NOT B & $A \land \lnot B$ & 0 & 0 & 1 & 0 \\
    $g_4$ & PROJ$_A$ & $A$ & 0 & 0 & 1 & 1 \\
    $g_5$ & NOT A AND B & $\lnot A \land B$ & 0 & 1 & 0 & 0 \\
    $g_6$ & PROJ$_B$ & $B$ & 0 & 1 & 0 & 1 \\
    $g_7$ & XOR & $A \oplus B$ & 0 & 1 & 1 & 0 \\
    $g_8$ & OR & $A \lor B$ & 0 & 1 & 1 & 1 \\
    $g_9$ & NOR & $\lnot(A \lor B)$ & 1 & 0 & 0 & 0 \\
    $g_{10}$ & XNOR & $A \leftrightarrow B$ & 1 & 0 & 0 & 1 \\
    $g_{11}$ & NOT B & $\lnot B$ & 1 & 0 & 1 & 0 \\
    $g_{12}$ & A OR NOT B & $A \lor \lnot B$ & 1 & 0 & 1 & 1 \\
    $g_{13}$ & NOT A & $\lnot A$ & 1 & 1 & 0 & 0 \\
    $g_{14}$ & NOT A OR B & $\lnot A \lor B$ & 1 & 1 & 0 & 1 \\
    $g_{15}$ & NAND & $\lnot(A \land B)$ & 1 & 1 & 1 & 0 \\
    $g_{16}$ & Constant $\operatorname{TRUE}$ & $1$ & 1 & 1 & 1 & 1 \\
    \bottomrule
    \end{tabular}
    \caption{Enumeration of all elementary Boolean gates; we denote the set $\mathbb{G}^{\star}\eqdef \{g_i\}_{i=1}^{16}$.}
    \label{tab:all_possible_gates}
\end{table}

\subsection{Elementary Properties of Boolean Gates}
\subsection{Boolean and Binary Logical Gates}
\label{s:Prelim__ss:BooleanLogic}

Note that there are two $0$-ary binary gates; namely the constant $\operatorname{TRUE}\eqdef 1$ or $\operatorname{FALSE}\eqdef 0$ gates.  Furthermore, there can only be two ``truly'' unary gates namely the identity gate sending any $x\in \{0,1\}$ to itself and the negation gate $\operatorname{NOT}$ sending $x$ to $1+x\,\mod(2)$.  Our first lemma observes that every $0$-ary and each unary gate can be realized as one of the binary gates in $\Bary$ upon composing with the projection $\pi:\{0,1\}^2\ni (x_1,x_2)\to x_1\in \{0,1\}$.
\begin{lemma}[{Embedding of Unary Gates}]
\label{lem:lift2binary}
There exist two constant binary gates; namely $g_1$ and $g_{16}$ in Table~\ref{tab:all_possible_gates}.  Moreover, let $g:\{0,1\} \to \{0,1\}$ be a unary Boolean gate.  Then, there exists a binary Boolean gate $g^{\uparrow}$ such that $g^{\uparrow}= g \circ \pi$. Consequently, we may study only binary Boolean gates and subsume our study of constant (0-ary) and unary Boolean gates.
\end{lemma}
\begin{proof}[{Proof of Lemma~\ref{lem:lift2binary}}]
By definition, the gates $g_1$ and $g_{16}$ in Table~\ref{tab:all_possible_gates} are constant.  We also have that $g_4=\operatorname{ID}\circ \, \pi$ and $g_{13}=\operatorname{NOT}\circ \, \pi$; the unary constant true and false gates are also given by $g_{16}$ and $g_{1}$ by composition with $\pi$; respectively.
\end{proof}
We denote the set of all standard Boolean gates $\Bool\eqdef \{\operatorname{AND},\operatorname{OR},\operatorname{NOT}\}$ as a subset of $\Bary$, however as the next result shows; there exists exactly $13$ other possible binary Boolean gates not included in $\Bool$ and these are listed in Table~\ref{tab:all_possible_gates}.
\begin{lemma}[{Enumeration of $\Bary$}]
\label{lem:Bool_Gates}
There exists exactly $16$, $2$-ary binary Boolean gates.  Moreover, these are enumerated in Table~\ref{tab:all_possible_gates}.
\end{lemma}
\begin{proof}[{Proof of Lemma~\ref{lem:Bool_Gates}}]
Every Boolean gate is a map $g:\{0,1\}^2\to \{0,1\}$.  Since $\#\{0,1\}=2$, then $\#\{0,1\}^2=4$; whence the cardinality of $\{0,1\}^{\{0,1\}^2}$ is $2^4=16$.  For the second statement, note that each gate in Table~\ref{tab:all_possible_gates} has a distinct truth assignment.
\end{proof}

\section{Boolean Neuronal Representability (BNR)}
\label{s:BNR_Appendix}
We now include a few more details on the BNR property of real-valued functions.

\begin{example}
\label{ex:Example}
Let $a>0$, the map $f:\{0,1\}^2\ni x \to a\cdot x\in \mathbb{R}^2$ is BNR.
\end{example}

\begin{example}
The map $f:\{0,1\}^2\ni x \to  (x_1+x_2,x_1)\in  \mathbb{R}^2$ is not BNR since $f((0,0))_1=0$, $f((0,1))_1=1$, and $f((1,1))_1=2$; meaning it is impossible to assign two distinct truth values to these three distinct values.
\end{example}

\begin{example}[ReLU neurons need not be BNR]
\label{ex:BnR}
Let $2\le B\in \mathbb{N}_+$.  Let $0<a_1<a_2$ and define $a\eqdef (a_i)_{i=1}^B$.  Thus
$
    \operatorname{ReLU}(a^{\top}0)= 0
<
    \operatorname{ReLU}(a^{\top}e_1)=a_1
<
    \operatorname{ReLU}(a^{\top}e_2)=a_2
$ 
so $f(\cdot)\eqdef \operatorname{ReLU}(a^{\top}\cdot + 0)$ fails Proposition~\ref{prop:BNR_Equivalence} (ii); whence $f$ is not BNR.
\end{example}

\begin{proposition}[Random ReLU neurons fail to be BNR w.h.p.]\label{thm:relu-not-bnr-w.h.p.}
Let $3 \le B\in\N_{+}$ and let $a\eqdef(a_i)_{i=1}^B\in\R^B$ have i.i.d.\ coordinates, where each $a_i$ is drawn from an atom-free distribution symmetric about $0$.
Define $f(\cdot)\eqdef \operatorname{ReLU}(a^\top \cdot + 0)$.
Then, with probability at least $1-\frac{B+2}{2^{B}}$,
 there exist \emph{distinct} indices $i,j,k\in[B]_+$ such that
\[
\operatorname{ReLU}(a^\top e_i) = 0
<
\operatorname{ReLU}(a^\top e_j)=a_j
<
\operatorname{ReLU}(a^\top e_k)=a_k,
\]
and hence $f(\cdot)$ fails Proposition~\ref{prop:BNR_Equivalence}~(ii); whence $f$ is not BNR.
\end{proposition}

\begin{proof}[Proof of Proposition \ref{thm:relu-not-bnr-w.h.p.}]
Since each $a_i$ is drawn i.i.d.\ from an atom-free distribution symmetric about $0$, we have
\[
\mathbb{P}(a_i>0)=\mathbb{P}(a_i<0)=\tfrac12,
\qquad
\mathbb{P}(a_i=0)=0,
\]
and moreover, by atom-freeness,
$\mathbb{P}(\exists\, i\neq j:\ a_i=a_j)=0$,
so that, almost surely, the values $a_1,\dots,a_B$ are pairwise distinct.

We have $a^\top e_i=a_i$ and hence $\operatorname{ReLU}(a^\top e_i)=\max\{a_i,0\}$.
Let $N_+\eqdef |\{i\in[B]_+:a_i>0\}|$, so $N_+\sim\mathrm{Binomial}(B,\tfrac12)$.
If $2\le N_+\le B-1$, then there exist distinct $i,j,k\in[B]_+$ with $a_i<0<a_j<a_k$
(note that $a_1,\dots,a_B$ are a.s.\ distinct), and therefore
\[
\operatorname{ReLU}(a^\top e_i)=0<\operatorname{ReLU}(a^\top e_j)=a_j<\operatorname{ReLU}(a^\top e_k)=a_k.
\]
Hence
\[
\mathbb{P}\bigl(\exists\, i,j,k\text{ distinct s.t.\ }0=\operatorname{ReLU}(a^\top e_i)
<\operatorname{ReLU}(a^\top e_j)
<\operatorname{ReLU}(a^\top e_k)\bigr)
\;\ge\;
\mathbb{P}(2\le N_+\le B-1).
\]
Finally,
\[
\mathbb{P}(2\le N_+\le B-1)
=
1-\mathbb{P}(N_+=0)-\mathbb{P}(N_+=1)-\mathbb{P}(N_+=B)
=
1-\frac{B+2}{2^B}.
\]
Thus, with probability at least $1-\frac{B+2}{2^B}$, the above inequalities hold. Thus, $f(\cdot)$ fails Proposition~\ref{prop:BNR_Equivalence} (ii), and hence $f(\cdot)$ is not BNR.
\end{proof}

The following provides a simple characterization of the BNR property which can be used as an efficient criterion for verifying if a neuron is Boolean neuronally representable (property (ii)) as canonical choice of a BNR.
\begin{proposition}[{Characterization of BNR - Extended Version of Proposition~\ref{prop:BNR_Equivalence}}]
\label{prop:BNR_Equivalence__EXTENDED}
Let $B_{in},B_{out}\in \mathbb{N}_+$ and $f:\{0,1\}^{B_{in}}\to \mathbb{R}^{B_{out}}$.  Then, the following are equivalent:
\begin{enumerate}
    \item[(i)] $f$ is BNR,
    \item[(ii)] For every $i\in [B_{out}]_+$, 
    $
        \#\{ f(x)_i : x\in\{0,1\}^{B_{in}}\}\le 2
    $,
    \item[(iii)] The map $\mathcal{R}_f:\mathbb{R}^{B_{out}}\to \{0,1\}^{B_{out}}$ sending any $x\in \mathbb{R}^{B_{out}}$ to 
    \[
            \mathcal{R}_f(x)
        \eqdef 
            \big(
                I(
                    x_i = f(\mathbf{0})_i
                )
            \big)_{i=1}^{B_{out}},
    \]
    where $\mathbf{0}\eqdef (0\,\dots\,0)\in\{0,1\}^{B_{in}}$,
    is a Boolean neuronal representation for $f$.
\end{enumerate}
We refer to $\mathcal{R}_f$ as the canonical BNR, whenever $f$ is BNR.
\end{proposition}
\begin{definition}[BNR Density]
\label{defn:BNR_Ratio}
Let $L\in\mathbb{N}_+$ and let $B_1,\dots,B_{L+1}\in \mathbb{N}_+$.  For each $l\in [L]_+$ let
$f_l:\mathbb{R}^{B_l}\to \mathbb{R}^{B_{l+1}}$ and consider the composite map
$f\eqdef f_L\circ \dots \circ f_1$.    
Define the sets $\mathcal{D}^f_0\eqdef \{0,1\}^{B_1}$ and, for every $l\in[L]_+$, iteratively define
\[
        \mathcal{D}^f_l
    \eqdef
        f_l(\mathcal{D}^f_{l-1}).
\]
The BNR density of $f$ is 
$
        \operatorname{BNR}(f)
    \eqdef 
        \sum_{l=1}^L\,
        \sum_{i=1}^{B_{l+1}}\,
            \frac{
                I\big(
                    \#\{ (f_l(x))_i : x\in \mathcal{D}^f_{l-1}\}\le 2
                \big)
            }{LB_{l+1}}
$.        
\end{definition}
\begin{example}[Boolean Functions (and Gates) Have BNR Density $1$]
\label{ex:Booleanfunctionstrivial_BNR}
Let $L=1$, $B_1=B_{in}$, and $B_2=B_{out}$.  For every
$f:\{0,1\}^{B_{in}}\to \{0,1\}^{B_{out}}$ we have $\operatorname{BNR}(f)=1$.
\end{example}

\subsection{Proofs Surrounding BNR}
\label{app:BNR_detail}
\begin{proof}[{Proof of Proposition~\ref{prop:BNR_Equivalence__EXTENDED}}]
Assume (i) and for every $i\in[B_{out}]_+$ write
\[
    S_i
\eqdef
    \{ f(x)_i : x\in\{0,1\}^{B_{in}} \}.
\]
If there exists some $i\in [B_{out}]_+$ such that $\#S_i\ge 3$, then the map
$\mathcal{R}_i:S_i\to\{0,1\}$ cannot be injective by the pigeonhole principle,
which contradicts the definition of BNR.  This proves (ii).  

\noindent
Next, assume (ii).  Fix $i\in[B_{out}]_+$.  Since $\#S_i\le 2$, there is at most one
element of $S_i$ different from $f(\mathbf{0})_i$.  Hence the predicate
$x_i=f(\mathbf{0})_i$ separates the elements of $S_i$, implying that
$(\mathcal{R}_f)_i\big|_{S_i}$ is injective.  Since $i$ was arbitrary, $\mathcal{R}_f$
is a Boolean neuronal representation for $f$, proving (iii).  

\noindent Finally, the implication (iii)$\Rightarrow$(i) is immediate.
\end{proof}

\begin{example}
The map $f:\{0,1\}^2\ni x \to  (x_1+x_2,x_1)\in  \mathbb{R}^2$ is not BNR since $f((0,0))_1=0$, $f((0,1))_1=1$, and $f((1,1))_1=2$; meaning it is impossible to assign two distinct truth values to these three distinct values.
\end{example}

The following provides a simple characterization of the BNR property which can be used as an efficient criterion for verifying if a neuron is Boolean neuronally representable (property (ii)) as canonical choice of a BNR.
\begin{proposition}[Characterization of BNR]
\label{prop:BNR_Equivalence}
Let $B_{in},B_{out}\in \mathbb{N}_+$ and $f:\{0,1\}^{B_{in}}\to \mathbb{R}^{B_{out}}$.  Then, the following are equivalent:
\begin{enumerate}
    \item[(i)] $f$ is BNR,
    \item[(ii)] For every $i\in [B_{out}]_+$, 
    $
        \#\{ f(x)_i : x\in\{0,1\}^{B_{in}}\}\le 2
    $.
\end{enumerate}
\end{proposition}

\section{Proofs}
\label{s:Proofs}

Our proof will follow a variant of the surgery approach of~\cite{kratsios2025quantifying}.  Namely, we first exhibit a worst-case Boolean circuit on the gates $\{\operatorname{AND},\operatorname{OR},\operatorname{PROJ}_B,\operatorname{NOT}\}$ which implement any given $S$-Junta.  We then swap the gates, this is the surgery technique of~\cite{kratsios2025quantifying}, for networks which perform the same computation but now with high probability.  Together, this will give universal reasoning.

\begin{proposition}\textbf{\emph{(Smooth Compactly Supported Exhaustion of  Elementary Boolean Gates)}}\label{prop:Enumeration_of_Fanin2_Fanout1__BooleanGates}
    For any $x\in \{0,1\}^2$ we have
    \[
            \sigma(x) 
        =
            \big(
                g_i(x)
            \big)_{i=1}^{16}
    \]
    where for each $i\in [16]_+$ $g_i$ is the elementary Boolean gate in Table~\ref{tab:all_possible_gates}.
    \hfill\\
    Additionally, $\sigma$ is ($C^{\infty}$-)smooth and compactly supported.
    \end{proposition}

    First, we confirm that any stochastic Boolean gate as in $\randactiv^W$ is almost surely a valid Boolean gate when sampled.
\begin{proof}[{Proof of Proposition~\ref{prop:Enumeration_of_Fanin2_Fanout1__BooleanGates}}]
    Since $\varsigma$ is a composition of the polynomial function $x\mapsto \|2x\|^2$ with the smooth compactly supported bump function $x\mapsto I_{|x|<1} e^{-\tfrac{1}{|x|^2-1}}$ then it is smooth.  Since $\tilde{\sigma}(\cdot|Z)$ is a linear combination of shifts of these maps so is each $\tilde{\sigma}(\cdot|Z)$ where $Z\in \{0,1\}^4$.  Since $\sigma$ is the direct sum of compactly supported smooth functions it is also smooth and compactly supported.

    By construction, for each $\bar{x}_i\in \{0,1\}^2$ and every $x\in \{0,1\}^2$ we have $\varsigma(x-\bar{x}_i)=1$ if and only if $x=\bar{x}_i$ and $\varsigma(x-\bar{x}_i)=0$ otherwise.  Consequently, for every $x\in \{0,1\}^2$ and each $i\in [4]_+$
    \begin{equation}
    \label{eq:interpolation_property}
            \tilde{\sigma}(x|Z) 
    = 
        \sum_{i=1}^4\, \varsigma(x-\bar{x}_i)Z_i
    =
        \sum_{i=1}^4 \, I_{x=\bar{x}_i}Z_i
    .
    \end{equation}
     By the interpolation identity in~\eqref{eq:interpolation_property} we have that: for each $i\in [4]_+$, $\tilde{\sigma}(x|Z)=Z_i$ if and only if $x=\bar{x}_i$.  
    \end{proof}

    \begin{proof}[{Proof of Proposition~\ref{prop:Valid_gate}}]
For any $W\in \mathbb{R}^{16}$, by definition of $\randactiv^W$, we have that $\randactiv^W \in \{\tilde{\sigma}(\cdot|Z_i)\}_{i\in [16]_+}$.  The result then follows from Proposition~\ref{prop:Enumeration_of_Fanin2_Fanout1__BooleanGates}.
\end{proof}

In the following proof, we use $\mathcal{P}(\mathcal{X})$ to denote the space of Borel probability measures on a metric space $\mathcal{X}$.  We metrize $\mathcal{P}(\mathcal{X})$ with the total variation metric $\operatorname{TV}$, which assigns to any probability measures $\mathbb{P},\mathbb{Q}\in \mathcal{P}(\mathcal{X})$ 
\[
        \operatorname{TV}(\mathbb{P},\mathbb{Q})
    \eqdef
        \sup_{A}\, 
            |\mathbb{P}(A)-\mathbb{Q}(A)|
\]
where the supremum is taken overall Borel subsets $A$ of $\mathcal{X}$.
We write $X\sim Y$ for two random variables, defined on the same probability space, whose laws coincide.  
\begin{proof}[{Proof of Proposition~\ref{prop:random_weight_assignment}}]
Fix $\eta>0$, to be specified retroactively.
By Proposition~\ref{prop:Enumeration_of_Fanin2_Fanout1__BooleanGates}, $\randactiv^W|_{\{0,1\}^2}\sim \mathbb{P}^W\eqdef \sum_{i=1}^{16}\, \operatorname{SM}(W)_i\delta_{g_i}$ where, $W=\eta e_i$. 
In particular, $\mathbb{P}^W\in \mathcal{P}\big(\{g_i\}_{i\in [16]_+}\big)$ equipped with the metric $\operatorname{TV}$ isometric to $(\Delta_{16},\frac{1}{2}\|\cdot\|_1)$ with isometry $\varphi:\Delta_{16}\to \mathcal{P}(\Bary)$ sending any $w\in \Delta_{16}$ to $\mathbb{P}_w\eqdef \sum_{i=1}^{16}\,w_i\,\delta_{g_i}$.  consequently, and by definition of $\mathbb{P}^W$ we have, $\operatorname{TV}(\mathbb{P}^W,\delta_{g_i})= \frac{1}{2}\|\operatorname{SM}(\eta e_i)-e_i\|_1$.
Now, since $\|\operatorname{SM}(\eta e_i)-e_i\|_1 = 
\frac{2(G-1)}{e^{\eta} + (G-1)}$; where $G=16$ then,
\begin{equation}
\label{eq:TV_Bound}
    \operatorname{TV}(\mathbb{P}^{\eta e_i},\delta_{g_i})
=
    \frac{1}{2}\|\operatorname{SM}(\eta e_i) -e_i\|_1
=
    \frac{15}{e^{\eta} + 15}
.
\end{equation}
Therefore, if $\eta \ge \ln\big(
\tfrac{15}{\delta} - 15
\big)$ then~\eqref{eq:TV_Bound} implies that $\operatorname{TV}(\mathbb{P}^{\eta e_i},\delta_{g_i})\le \delta$.  By definition of the total variation distance, we have
\[
        \mathbb{P}\big(
                \randactiv^{\eta e_i}|_{\{0,1\}^2}  = g_i
        \big)
    =
        1-\operatorname{TV}(\mathbb{P}^{\eta e_i},\delta_{g_i})
    =
        1-\frac{1}{2}\|\operatorname{SM}(\eta e_i) -e_i\|_1
    \ge 
        1-\delta
\]
completing our proof.
\end{proof}

\begin{example}
\label{ex:cube_sampling}
Let $B=2$ and $B^{\uparrow}=3$.  Fix any $x\in\{0,1\}^{2}$.
For every $\eta>0$, define $W^{\eta}\in\R^{B^{\uparrow}\times 2B}$ by $
W^\eta \;=\; \eta\,\big(e_1,\,e_3,\,e_2\big)^\top,
$
where $(e_1,e_2,e_3,e_4)$ denotes the standard basis of $\R^{4}$.
Then, as $\eta\to\infty$, the random variables $(\mathbf{X}^{\uparrow:W^{\eta}})_{\eta>0}$ tends to the degenerate random variable taking value $
(x_1,\neg x_1,x_2)$ $\mathbb{P}$-a.s.
\end{example}
\begin{proof}[{Proof of Proposition~\ref{prop:lifting_channels}}]
The first claim is trivial.  For the second, the argument is nearly analogous to the proof of Proposition~\ref{prop:random_weight_assignment}, mutatis mutandis.
\end{proof}

\begin{proof}[{Proof of Proposition~\ref{prop:Good_Properties_Uncrossable}}]
Now observe that Property (i) holds since the range of $\operatorname{SM}$ is $\dot{\Delta}_B$.

Property (ii) holds since $\operatorname{SM}$ is analytic, as well as the Hadamard product, multiplication and addition by constants; as well as the composition and concatenation of ($C^{\infty}$-)smooth functions.

To show (iii), we first prove two properties of $\operatorname{SM}$. First we will show that $\operatorname{SM}$ is monotone. That is, given $W \in \mathbb{R}^B$ with $(W)_i < (W)_j$ for $i, j \in B$ and $\eta > 0$,
\begin{equation}
\label{eq:part-i-lemma-1}
\operatorname{SM}(\eta W)_i < \operatorname{SM}(\eta W)_j.
\end{equation}
This follows from the definition of $\operatorname{SM}$ and the monotonicity of $e$:
\begin{equation}
\operatorname{SM}(\eta W)_i = \frac{e^{\eta (W)_i}}{\sum_{k=1}^{B} e^{\eta (W)_k}} < \frac{e^{\eta (W)_j}}{\sum_{k=1}^{B} e^{\eta (W)_k}} = \operatorname{SM}(\eta W)_j.
\end{equation}

Towards the second property, assume again that $(W)_i < (W)_j$. Then we will show that for all $\varepsilon > 0$, there exists some $\eta > 0$ such that
\begin{equation}
\label{eq:part-i-lemma-2}
\operatorname{SM}(\eta W)_i < \epsilon.
\end{equation}

Now it follows by the definition of $\operatorname{SM}$ that
\begin{equation}
    \operatorname{SM}(\eta W)_i = \frac{e^{\eta (W)_i}}{\sum_{k=1}^{B} e^{\eta (W)_k}}
    < \frac{e^{\eta (W)_i}}{e^{\eta (W)_j}} = e^{-\eta((W)_j - (W)_i)}.
\end{equation}
Setting $e^{-\eta((W)_j - (W)_i)} < \varepsilon$ yields
\begin{equation}
\label{eq:part-i-lemma-2-result}
\eta > \frac{\ln{(1/\varepsilon)}}{(W)_j - (W)_i}.
\end{equation}

Now towards (iii): Notice that the definition of $\operatorname{SM}$ and (\ref{eq:part-i-lemma-1}) necessitate that
\begin{equation}
    \eta(1 - \operatorname{SM}(\eta e_i)_i)\operatorname{SM}(\eta (W_2))_i < \eta(1 - \operatorname{SM}(\eta e_i)_j)\operatorname{SM}(\eta (W_2))_j,
\end{equation}
and so by (\ref{eq:part-i-lemma-2}), given $\varepsilon > 0$, there exists $\eta^{\star:1} > 0$ such that for all $\eta > \eta^{\star:1}$ the claim holds.

We now prove (iv). For every $k\in [B]_+\setminus\{i\}$,
\allowdisplaybreaks
\begin{align}
\label{eq:remove_k}
    \lim\limits_{\eta\to \infty}
    \,
        \operatorname{SM}(\eta e_i)_k 
    = 
    \lim\limits_{\eta\to \infty}
    \,
    \frac{
            1
        }{
            (B-1)e^{\eta 0} + e^{\eta}
        }
    =
    \lim\limits_{\eta\to \infty}
    \,
    \frac{
            1
        }{
            (B-1) + e^{\eta}
        }
    =
        0
.
\end{align}
Since the range of $\operatorname{SM}$ belongs to the $B$-simplex then
\begin{equation}
\label{eq:remove_1}
        \lim\limits_{\eta\to \infty}
    \,
        \operatorname{SM}(\eta e_i)_i
    =
    \lim\limits_{\eta\to \infty}
    \,
        1-\sum_{k>1}^B\,
            \operatorname{SM}(\eta e_i)_k 
    =
    \lim\limits_{\eta\to \infty}
    \,
        1-0
    =
        1
.
\end{equation}
Together~\eqref{eq:remove_1} and~\eqref{eq:remove_k} imply that there exists some $\eta^{\star:2}>0$ such that: for all $\eta\ge \eta^{\star:2}$ we have
\begin{equation}
\label{eq:Conv1}
        \|\operatorname{SM}(\eta e_i) - e_i\|_1 
    < 
        \tfrac{\varepsilon}{2}
.
\end{equation}
Now take $w=e_j$ such that $j \not= i$. Simplifying
$\operatorname{SM}\Big(\eta 
    \big(
            (1-\operatorname{SM}(\eta e_i))
        \odot 
            \operatorname{SM}(\eta e_j)
    \big)
    \Big)_j$ we find that it equals to
\allowdisplaybreaks
\begin{align}
\label{eq:ugly_expression_second_coordinate}
    \frac{
            \exp\Big( 
                \eta \frac{e^\eta (e^\eta + B - 2)}{(e^\eta + (B-1))^2} 
            \Big)
        }{
            \exp\Big( 
                \eta \frac{B-1}{(e^\eta + (B-1))^2} 
            \Big)
        + 
            \exp\Big( 
                \eta \frac{e^\eta (e^\eta + B - 2)}{(e^\eta + (B-1))^2} 
            \Big)
        + (B-2) 
            \exp\Big( \eta \frac{e^\eta + B - 2}{(e^\eta + (B-1))^2} \Big)
        }
.
\end{align}
Now, taking the limit of~\eqref{eq:ugly_expression_second_coordinate} as $\eta$ tends to infinity, we find that
\begin{align}
\label{eq:ugly_expression_second_coordinate2}
&
\lim\limits_{\eta \to \infty}\,
    \operatorname{SM}\Big(\eta 
    \big(
            (1-\operatorname{SM}(\eta e_i))
        \odot 
            \operatorname{SM}(\eta e_j)
    \big)
    \Big)_j
\\
=
&
\lim\limits_{\eta \to \infty}\,
    \frac{
            \exp\Big( 
                \eta \frac{e^\eta (e^\eta + B - 2)}{(e^\eta + (B-1))^2} 
            \Big)
        }{
            \exp\Big( 
                \eta \frac{B-1}{(e^\eta + (B-1))^2} 
            \Big)
        + 
            \exp\Big( 
                \eta \frac{e^\eta (e^\eta + B - 2)}{(e^\eta + (B-1))^2} 
            \Big)
        + (B-2) 
            \exp\Big( \eta \frac{e^\eta + B - 2}{(e^\eta + (B-1))^2} \Big)
        }
= 1
.
\end{align}
Since $\operatorname{SM}$ is continuous and maps into $\dot{\Delta}_B$ which is a closed set; then its limit must belong to the $B$-simplex also.  However, if 
$
\lim\limits_{\eta \to \infty}\,
    \operatorname{SM}\Big(\eta 
    \big(
            (1-\operatorname{SM}(\eta e_i))
        \odot 
            \operatorname{SM}(\eta e_j)
    \big)
    \Big)_j$ and each coordinate of the limit is non-negative, we have that
\[
    \operatorname{SM}\Big(\eta 
    \big(
            (1-\operatorname{SM}(\eta e_i))
        \odot 
            \operatorname{SM}(\eta e_j)
    \big)
    \Big)_i
\]
for every $i\in [B]_+\setminus\{j\}$.  Therefore, for all $\varepsilon > 0$, there exists some $\eta^{\star:3}>0$ such that, if $\eta>\eta^{\star}\eqdef \max\{\eta^{\star:1},\eta^{\star:2},\eta^{\star:3}\}$ we have
\begin{equation}
\label{eq:second_coordinate_handled__eq:Conv2}
        \big\|
            \operatorname{SM}\Big(\eta 
            \big(
                    (1-\operatorname{SM}(\eta e_i))
                \odot 
                    \operatorname{SM}(\eta e_j)
            \big)
            \Big)
        -
            e_j
        \big\|_1
<
    \tfrac{\varepsilon}{2}
.
\end{equation}
Combining~\eqref{eq:Conv1} and~\eqref{eq:second_coordinate_handled__eq:Conv2} yields (ii) and completes our proof.

Let $(\mathbf{E}_{i:1},\mathbf{E}_{i:2})\sim S_{\eta}(e_i,e_j)$.
Define the disagreement set $A^{\star}$ as:
\[A^{\star}\eqdef \{\omega \in \Omega:\, 
(\mathbf{E}_{i:1}(\omega),\mathbf{E}_{i:2}(\omega))
\neq (e_i,e_j)
\}.\] 
Note that $A^{\star}$ is measurable since $\mathbf{E}_{i:j}$ are for $j\in \{1,2\}$ and since $\mathbb{R}^2\ni (x,y)\mapsto I_{x=y}\in \{0,1\}$ is Borel. 
Now, for the last statement, we have by the dual representation of the TV-distance and by (ii)
\allowdisplaybreaks
\begin{align}
\label{eq:little_fun_bound}
        \big\|
            \operatorname{Law}\big((\mathbf{E}_{i:1},\mathbf{E}_{i:2})\big)(A^{\star})
            -
            \delta_{(e_i,e_j)}(A^{\star})
        \big\|
    & \le 
        \sup_{A}\,
            \big\|
                \operatorname{Law}\big((\mathbf{E}_{i:1},\mathbf{E}_{i:2})\big)(A)
                -
                \delta_{(e_i,e_j)}(A)
            \big\|
    \\
\nonumber
    & =
        \big\|
            \operatorname{Law}\big((\mathbf{E}_{i:1},\mathbf{E}_{i:2})\big)
            -
            \delta_{(e_i,e_j)}
        \big\|_{TV}
    \\
\nonumber
    & =
        \big\|
            (\mathbf{E}_{i:1},\mathbf{E}_{i:2})
            -
            (e_i,e_j)
        \big\|_{1}
    \le 
        \varepsilon
\end{align}
where the supremum is taken over all measurable sets $A$.  Consequently, taking complements across~\eqref{eq:little_fun_bound} yields the final claim.
\end{proof}

\begin{proof}[{Proof of Theorem~\ref{thrm:StructureAware}}]
We recall that the recursion in \eqref{eq:Model_Def} is indexed by
\(
l = 1,\dots,\Delta-1,
\)
so that for each \(l\in\{1,\dots,\Delta-1\}\) the layer representations satisfy
\(
\boldsymbol{x}^{(l)} \in \{0,1\}^{B_l}
\)
and
\(
\boldsymbol{x}^{(l+1)} \in \{0,1\}^{B_{l+1}}.
\)
Accordingly, for each \(l \in \{1,\dots,\Delta-1\}\) and \(j \in \{1,2\}\), the random adjacency matrix
\(
\randEdge^{(l:j:e)}
\)
is a \(B_{l+1} \times B_l\) matrix whose rows are one-hot almost surely; equivalently,
\[
\randEdge^{(l:j:e)}_{i:} \in \{e_1,\dots,e_{B_l}\}
\quad \forall i \in [B_{l+1}]_+, \ \mathbb{P}\text{-a.s.}
\]
In contrast, the transition from the raw input layer \(l=0\) to the lifted layer \(l=1\)
is produced by the stochastic bit-lifting channel \(\mathbf{X}^{\uparrow:W}\) (rather than by the matrices
\(\randEdge^{(l:j:e)}\)).

Fix $\Delta\in\mathbb{N}_+$, $\Upsilon\in\mathbb{N}_+$, and widths
$B_{in}=B_0,B_1,\dots,B_\Delta=B_{out}=1$ with $B_l\le \Upsilon$ for $l \in [\Delta]$.
Let $\randcirc$ be defined by \eqref{eq:Model_Def}.
We show that, $\mathbb{P}$-a.s.\ in $\omega$, there exists
$\mathcal{C}_\omega \in \operatorname{CIRCUIT}_{\Delta,\Upsilon}^{B_{in}}$
such that $\randcirc(\omega,x)=f_{\mathcal{C}_\omega}(x)$ for all
$x\in\{0,1\}^{B_{in}}$.

\paragraph{Step 1: A full-probability event on which all modules are Boolean.}

All random objects in the model are defined on a fixed probability space
$(\Omega,\mathcal{F},\mathbb{P})$.
In particular:
the lifted-input map takes values in the finite set $\{0,1\}^{B_1}$. Also,
for each \(l \in \{1,\dots,\Delta-1\}\) and \(j \in \{1,2\}\), each row of every adjacency matrix \(
\randEdge^{(l:j:e)}
\) takes values in the finite set
$\{e_1,\dots,e_{B_l}\}$,
and each gate restriction belongs to the finite set $\Bary$.
Hence all random variables involved in the definitions below are finite-valued,
and any event defined by constraints on their values is $\mathcal{F}$-measurable.

\smallskip

Let $\Omega_{\mathrm{lift}}\subseteq\Omega$ be the event from
Proposition~\ref{prop:lifting_channels} such that for every
$\omega\in\Omega_{\mathrm{lift}}$ and every input $x\in\{0,1\}^{B_{in}}$,
\[
\mathbf{X}^{\uparrow:W}(\omega)(x)\in\{0,1\}^{B_1}.
\]
Since the lifted-input map is finite-valued, $\Omega_{\mathrm{lift}}\in\mathcal{F}$,
and Proposition~\ref{prop:lifting_channels} implies that
$\mathbb{P}(\Omega_{\mathrm{lift}})=1$.

\smallskip

Let $\Omega_{\mathrm{gate}}\subseteq\Omega$ be the event from
Proposition~\ref{prop:Valid_gate} such that for every
$\omega\in\Omega_{\mathrm{gate}}$, every $l\in\{1,\dots,\Delta-1\}$,
and every node index $j\in\{1,\dots,B_{l+1}\}$,
the restriction of the sampled activation map
\[
\randactiv^{W^{(l:\sigma)}_j}(\omega)\big|_{\{0,1\}^2}
\]
belongs to $\Bary$, i.e., defines a meaningful Boolean gate on $\{0,1\}^2$.
Since this restriction takes values in the finite set $\Bary$,
the event $\Omega_{\mathrm{gate}}$ is $\mathcal{F}$-measurable.
By Proposition~\ref{prop:Valid_gate}, we have
$\mathbb{P}(\Omega_{\mathrm{gate}})=1$.

\smallskip

Finally, for each $l\in\{1,\dots,\Delta-1\}$, the random adjacency matrices
$\randEdge^{(l:1:e)}$ and $\randEdge^{(l:2:e)}$ select exactly one predecessor
in layer $l$ for every node in layer $l+1$, i.e., each row is one-hot.
Let $\Omega_{\mathrm{edge}}\subseteq\Omega$ be the event on which, for every
$l\in\{1,\dots,\Delta-1\}$, $j\in\{1,2\}$, and every $i\in[B_{l+1}]_+$,
\[
\randEdge^{(l:j:e)}_{i:}(\omega)\in\{e_1,\dots,e_{B_l}\}.
\]
Since each adjacency row takes values in a finite set,
$\Omega_{\mathrm{edge}}\in\mathcal{F}$.
By the assumptions on the adjacency matrices,
$\mathbb{P}(\Omega_{\mathrm{edge}})=1$.

\smallskip

Define
\[
\Omega_{\mathrm{all}}
\eqdef
\Omega_{\mathrm{lift}}
\cap
\Omega_{\mathrm{gate}}
\cap
\Omega_{\mathrm{edge}}.
\]
Then $\Omega_{\mathrm{all}}\in\mathcal{F}$ and
$\mathbb{P}(\Omega_{\mathrm{all}})=1$.
Fix any $\omega\in\Omega_{\mathrm{all}}$ for the remainder of the proof.
We construct a circuit
$\mathcal{C}_\omega\in\operatorname{CIRCUIT}_{\Delta,\Upsilon}^{B_{in}}$
computing $x\mapsto \randcirc(\omega,x)$.

\paragraph{Step 2: Construct the realized circuit $\mathcal{C}_\omega$.}

We build a layered DAG with layers
\begin{equation}\label{eqn:DAG_nodes}
    V_0=\{v_{0,1},\dots,v_{0,B_0}\},\ 
V_1=\{v_{1,1},\dots,v_{1,B_1}\},\ \dots,\
V_{\Delta}=\{v_{\Delta,1}\},
\end{equation}

so that layer $V_0$ corresponds to the input $x \in \{0,1\}^{B_0}$, and for $l \in \{1,\dots,\Delta\}$ layer $V_{l}$ corresponds to the model state $\boldsymbol{x}^{(l)}\in\{0,1\}^{B_l}$.

Let $\operatorname{val}_{\mathcal{C}_\omega}(v_{l,i};x)$, for $l\in[\Delta]$ and $i\in [B_l]_+$, denote the output value of the $i$-th node in layer $l$ of the layered DAG, evaluated at the input $x$, in the induced Boolean function $f_{\mathcal{C}_\omega}$ (corresponding to the node $v_{l,i}$ as in~\eqref{eqn:DAG_nodes}).
 In particular, $\operatorname{val}_{\mathcal{C}_\omega}(v_{\Delta,1};x)$ corresponds to the output value of the induced Boolean function $f_{\mathcal{C_{\omega}}}(x)=f_\Delta(x)$ as in~\eqref{eq:recursive_defenition}. We next define these values starting from the input layer.

\smallskip
\noindent\textbf{Input values.}
For \(x\in\{0,1\}^{B_{in}}\), define the value at the raw input nodes by
\begin{equation}\label{eqn:raw_input_layer_value}
\operatorname{val}_{\mathcal{C}_\omega}(v_{0,i};x)\eqdef x_i
\qquad \forall i\in\{1,\dots,B_0\}.
\end{equation}

\smallskip
\noindent\textbf{Lifted-layer values.}
For \(x\in\{0,1\}^{B_{in}}\), define the value at \(v_{1,j}\) to be
\begin{equation}\label{eqn:input_layer_value}
\operatorname{val}_{\mathcal{C}_\omega}(v_{1,j};x)\eqdef \boldsymbol{x}^{(1)}_j(\omega,x)
= \mathbf{X}^{\uparrow:W}(\omega)(x)_j\in\{0,1\} \qquad \forall j \in \{1,\cdots,B_1\}, 
\end{equation}
using $\omega\in\Omega_{\mathrm{lift}}$.

\smallskip
\noindent\textbf{Lifting edges (layer \(0\to 1\)).}
Fix \(\omega\in\Omega_{\mathrm{lift}}\).
By definition of the lifting channel \(\mathbf{X}^{\uparrow:W}\),
there exists a (random) binary matrix
\(
W(\omega)\in\{0,1\}^{B_1\times 2B_0}
\)
such that, for every input \(x\in\{0,1\}^{B_0}\),
the lifted representation \(\boldsymbol{x}^{(1)}(\omega,x)\) is obtained
by selecting coordinates of the extended input vector $\begin{pmatrix}x\\ \mathbf{1}_{B_0}-x\end{pmatrix}\in\{0,1\}^{2B_0}$.
In particular, for each \(j\in\{1,\dots,B_1\}\),
the \(j\)-th row of \(W(\omega)\) selects a unique index
\(q(j;\omega)\in\{1,\dots,2B_0\}\).

Accordingly, define for each \(j\in\{1,\dots,B_1\}\) the associated raw index
\(i(j;\omega)\in\{1,\dots,B_0\}\) 
by
\[
i(j;\omega)\eqdef
\begin{cases}
q(j;\omega), & q(j;\omega)\le B_0,\\
q(j;\omega)-B_0, & q(j;\omega)>B_0.
\end{cases}
\]
We then record the underlying raw coordinate by adding the directed edge
\[
\bigl(v_{0,i(j;\omega)},\, v_{1,j}\bigr).
\]
Here the edge identifies the raw index \(i(j;\omega)\). 

Note that the values at layer \(V_1\) are still defined by \eqref{eqn:input_layer_value};
the edges from \(V_0\) to \(V_1\) encode the correspondence between lifted bits and
raw input coordinates induced by \(W(\omega)\).

\smallskip
\noindent\textbf{Edges.}
For each $l\in\{1,\dots,\Delta-1\}$ and each node index $j\in\{1,\dots,B_{l+1}\}$, since
$\omega\in\Omega_{\mathrm{edge}}$, the $j$-th row of $\randEdge^{(l:1:e)}(\omega)$ is one-hot,
so it selects a unique index $p_1(l,j;\omega)\in\{1,\dots,B_l\}$; similarly
the $j$-th row of $\randEdge^{(l:2:e)}(\omega)$ selects a unique
$p_2(l,j;\omega)\in\{1,\dots,B_l\}$, 
Add directed edges
\[
\bigl(v_{l,p_1(l,j;\omega)},\, v_{l+1,j}\bigr),\qquad
\bigl(v_{l,p_2(l,j;\omega)},\, v_{l+1,j}\bigr).
\]
Thus, every node in layers \(V_l\) for \(l\ge 2\) has fan-in exactly \(2\), edges occur only between consecutive layers, and the edges from \(V_0\) to \(V_1\) encode the correspondence induced by the lifting map.

\smallskip
\noindent\textbf{Gates.}
For each $l\in\{1,\dots,\Delta-1\}$ and each $j\in\{1,\dots,B_{l+1}\}$, assign to node $v_{l+1,j}$
the Boolean gate
\[
g_{l+1,j}\eqdef \randactiv^{W^{(l:\sigma)}_j}(\omega)\big|_{\{0,1\}^2}\in\Bary,
\]
using $\omega\in\Omega_{\mathrm{gate}}$.

\smallskip
\noindent\textbf{Definition of the realized circuit.}
The realized Boolean circuit associated with $\omega\in\Omega_{\mathrm{all}}$ is the triple
$
\mathcal{C}_\omega \eqdef (V,E,g_\cdot),
$
where:
\begin{itemize}
\item The vertex set is
$V \eqdef \bigsqcup_{l=0}^{\Delta} V_l$ with $V_l$ defined as in \eqref{eqn:DAG_nodes}.

\item The edge set is
\[
E \eqdef
\Big(
\bigcup_{j=1}^{B_1} \{(v_{0,i(j;\omega)},\,v_{1,j})\}
\Big)
\ \cup\
\Big(
\bigcup_{l=1}^{\Delta-1}\;
\bigcup_{j=1}^{B_{l+1}}
\bigl\{
(v_{l,p_1(l,j;\omega)},\,v_{l+1,j}),
(v_{l,p_2(l,j;\omega)},\,v_{l+1,j})
\bigr\}
\Big).
\]

\item The gate assignment is
$
g_\cdot \eqdef
\{\,g_{l+1,j}:\{0,1\}^2\to\{0,1\}\,\}_{l=1,\dots,\Delta-1;\,j=1,\dots,B_{l+1}}.
$
\end{itemize}
This defines a Boolean circuit $\mathcal{C}_\omega=(V,E,g_\cdot)$ of depth $\Delta$ and width at most
$\max_{0\le l\le \Delta} B_l\le \Upsilon$.

\paragraph{Step 3: Prove that $\mathcal{C}_\omega$ computes $\randcirc(\omega,\cdot)$ (induction on depth).}

We prove that for every $l\in\{1,\dots,\Delta\}$, every input
$x\in\{0,1\}^{B_{in}}$, and every node index $j\in\{1,\dots,B_l\}$,
\begin{equation}\label{eq:induction_invariant}
\operatorname{val}_{\mathcal{C}_\omega}(v_{l,j};x)
=
\boldsymbol{x}^{(l)}_j(\omega,x).
\end{equation}

\medskip
\noindent\textbf{Base case ($l=1$).}
By definition of the input layer and the lifted-input channel,
\[
\operatorname{val}_{\mathcal{C}_\omega}(v_{1,j};x)\eqdef \boldsymbol{x}^{(1)}_j(\omega,x)
= \mathbf{X}^{\uparrow:W}(\omega)(x)_j\in\{0,1\} \qquad \forall j \in \{1,\cdots,B_1\}
\]
Hence, \eqref{eq:induction_invariant} holds for $l=1$ (see \eqref{eqn:input_layer_value}).

\medskip
\noindent\textbf{Induction hypothesis.}
Assume that for some $l\in\{1,\dots,\Delta-1\}$,
\eqref{eq:induction_invariant} holds for all nodes in layer $l$, i.e.
\[
\operatorname{val}_{\mathcal{C}_\omega}(v_{l,j};x)
=
\boldsymbol{x}^{(l)}_j(\omega,x)
\quad
\forall j\in\{1,\dots,B_l\},\ \forall x\in\{0,1\}^{B_{in}}.
\]

\medskip
\noindent\textbf{Inductive step.}
Fix an arbitrary node index $j\in\{1,\dots,B_{l+1}\}$ and input
$x\in\{0,1\}^{B_{in}}$.
By construction of $\mathcal{C}_\omega$,
\[
\operatorname{val}_{\mathcal{C}_\omega}(v_{l+1,j};x)
=
g_{l+1,j}\Big(
\operatorname{val}_{\mathcal{C}_\omega}(v_{l,p_1(l,j;\omega)};x),
\operatorname{val}_{\mathcal{C}_\omega}(v_{l,p_2(l,j;\omega)};x)
\Big).
\]
By the induction hypothesis,
\[
\operatorname{val}_{\mathcal{C}_\omega}(v_{l,p_i(l,j;\omega)};x)
=
\boldsymbol{x}^{(l)}_{p_i(l,j;\omega)}(\omega,x),
\qquad i=1,2.
\]
Since
$
g_{l+1,j}
=
\randactiv^{W^{(l:\sigma)}_j}(\omega)\big|_{\{0,1\}^2}
$
and the one-hot rows of
$\randEdge^{(l:1:e)}(\omega)$ and $\randEdge^{(l:2:e)}(\omega)$
select exactly the predecessors $p_1(l,j;\omega)$ and $p_2(l,j;\omega)$,
the recursion \eqref{eq:Model_Def} implies
\[
\boldsymbol{x}^{(l+1)}_j(\omega,x)
=
\randactiv^{W^{(l:\sigma)}_j}(\omega)\Big(
\boldsymbol{x}^{(l)}_{p_1(l,j;\omega)}(\omega,x),
\boldsymbol{x}^{(l)}_{p_2(l,j;\omega)}(\omega,x)
\Big).
\]
Thus
\[
\operatorname{val}_{\mathcal{C}_\omega}(v_{l+1,j};x)
=
\boldsymbol{x}^{(l+1)}_j(\omega,x).
\]

Since $j$ was arbitrary, \eqref{eq:induction_invariant} holds for all nodes
in layer $l+1$.
This completes the induction.

\paragraph{Step 4: Final conclusion.}
Taking $l=\Delta$ in \eqref{eq:induction_invariant} and recalling $B_\Delta=1$, we have for every input
$x\in\{0,1\}^{B_{in}}$,
\[
\randcirc(\omega,x)=\boldsymbol{x}^{(\Delta)}(\omega,x)
=
\operatorname{val}_{\mathcal{C}_\omega}(v_{\Delta,1};x),
\]
so $\randcirc(\omega,\cdot)$ is computed by the Boolean circuit $\mathcal{C}_\omega$ of depth $\Delta$  and width $\le \Upsilon$.
Therefore $\randcirc(\omega,\cdot)\in \operatorname{CIRCUIT}_{\Delta,\Upsilon}^{B_{in}}$ for every $\omega\in\Omega_{\mathrm{all}}$.
Since $\mathbb{P}(\Omega_{\mathrm{all}})=1$, it follows that
\[
\mathbb{P}\big(\randcirc\in \operatorname{CIRCUIT}_{\Delta,\Upsilon}^{B_{in}}\big)=1.
\]
This completes the proof.
\end{proof}

\subsection{Tree Construction}
\label{s:Proof__ss:TreeConstruction}
We begin with the worst-case representation of a Boolean function as a tree whose leaves are literals in the Boolean cube, their negations, and possibly repeated copies of each. Moreover, the circuit is implemented as a full binary tree, so it has the exact desired multipartite structure: every internal node is a fan-in-$2$, fan-out-$1$ gate of type $\operatorname{AND}$, $\operatorname{OR}$, or $\operatorname{PROJ}_A$. This avoids the complications that arise with circuit representations whose computational graph differs from that of our neural networks, or which mix unary and binary gates at internal nodes and therefore yield an incompatible computational structure.

\begin{lemma}[Complete Binary Tree Representation of Boolean Functions]
\label{lem:complete_binary_tree_representation}

Let $B\in \mathbb{N}_+$ and let $f:\{0,1\}^B\to\{0,1\}$ be a Boolean function.
There exist integers $B^{\uparrow},\Delta\in \mathbb{N}_+$ and a Boolean circuit $\mathcal{C}$ computing $f$ such that:
\begin{enumerate}[leftmargin=2em]
    \item[(i)] the underlying DAG of $\mathcal{C}$ is a rooted complete binary tree of depth $\Delta$,
    \item[(ii)] the $B^{\uparrow}$ leaves are labelled by literals in $\{x_1,\dots,x_B,\neg x_1,\dots,\neg x_B\}$, with repetitions allowed; i.e.\ for every $\tilde{b}\in [B^{\uparrow}]_+$ we have
    \[
        x_{\tilde{b}} 
        \in 
        \{x_b\}_{b=1}^B
        \cup
        \{\neg x_b\}_{b=1}^B
    ,
    \]
    \item[(iii)] every gate in $\mathcal{C}$ belongs to $\{\operatorname{OR},\operatorname{AND},\operatorname{PROJ}_A\}=\{g_8,g_2,g_4\}$.
\end{enumerate}
Moreover, 
the number of gates is exactly $B^{\uparrow}-1\le 2 B 2^B-1$
 and 
\begin{equation}
\label{eq:boundsboundybounds}
    B^{\uparrow}  \le  2 B 2^B,
    \qquad
    \Delta  \le  \big\lceil\log_2(B 2^B)\big\rceil  \le  B+\big\lceil\log_2 B\big\rceil
.
\end{equation}
\end{lemma}

We now sharpen the previous result to functions which are \textit{possibly} ``sparse'' in a Boolean sense.
\begin{definition}[$S$-Junta]
\label{defn:sJunta}
A Boolean function $f:\{0,1\}^B\to\{0,1\}$ is said to \emph{depend on the $i$th variable} if there exist inputs 
$x,y\in\{0,1\}^B$ which differ only in their $i$th coordinate and satisfy $f(x)\neq f(y)$. 
Equivalently, we say that the $i$th variable is \emph{relevant} to $f$.  
If $f$ has at most $s$ relevant variables, then we call $f$ an \emph{$s$-junta}.  
\end{definition}
Clearly a $B$-Junta on $B$ bits is just an arbitrary Boolean function on $\{0,1\}^B$.  However, at the other extreme of the spectrum we have dictator functions; defined as follows.
\begin{example}[Dictator Functions]
Let $B\in \mathbb{N}_+$.  A Boolean function $f:\{0,1\}^B\to \{0,1\}$ is called a dictator if there exists some $b\in [B]_+$ such that: either $f(x)=x_b$ or $f(x)=\neg x_b$ for every $x\in \{0,1\}^B$.  Dictator functions are precisely $1$-juntas.
\end{example}

The following equivalent formulation of a junta will be more handy for us.
\begin{lemma}[$S$-Juntas as Embedded Permutations]
\label{lem:sJunta_as_embeddedpermutation}
Let $B,S\in\mathbb{N}_+$ with $S\le B$.  
A Boolean function $f:\{0,1\}^B\to\{0,1\}$ is an $S$-junta iff there exist distinct 
$(i_1,\dots,i_S)\in[B]_+$ and some $f_S:\{0,1\}^S\to\{0,1\}$ such that
\[
    f(x) = f_S\bigl(x_{i_1},\dots,x_{i_S}\bigr)
    \quad\text{for all }x\in\{0,1\}^B.
\]
Equivalently, letting $\Pi\in\{0,1\}^{S\times B}$ be the selection matrix with
$ 
(\Pi x)_j = x_{i_j}
$ for each $j\in [S]_+$ and for every $x\in \{0,1\}^B$ we have
\begin{equation}
    f(x) = f_S(\Pi x)
.
\end{equation}
\end{lemma}
A direct consequence of Lemmata~\ref{lem:complete_binary_tree_representation} and~\ref{lem:sJunta_as_embeddedpermutation} imply the following sharper result.
\begin{lemma}\textbf{\emph{(Construction of Low-Dimensional Complete Boolean Circuit on Lifted Literals)}}
\label{lem:JuntaTreeification}
Let $S,B\in \mathbb{N}_+$ with $S\le B$, and let $f:\{0,1\}^B\to \{0,1\}$ be an $S$-Junta. 
Then, there exists a $B^{\uparrow}\in \mathbb{N}_+$, a stochastic matrix $\Pi\in \{0,1\}^{B^{\uparrow} \times B}$, a Boolean function $f_S:\{0,1\}^{B^{\uparrow}}\to \{0,1\}$, integers $B^{\uparrow},\Delta\in \mathbb{N}_+$, and a Boolean circuit $\mathcal{C}$ such that
\begin{equation}
\label{eq:dim_reduced_compute}
    f = f_S\circ (\Pi\,\cdot)
\end{equation}
for every $x\in \{0,1\}^B$ and the following hold: 
\begin{enumerate}[leftmargin=2em]
    \item[(i)] $\mathcal{C}$ computes $f_S$,
    \item[(ii)] the underlying DAG of $\mathcal{C}$ is a rooted complete binary tree of depth $\Delta$,
    \item[(iii)] the $B^{\uparrow}$ leaves are labelled by literals in $\{x_1,\dots,x_B,\neg x_1,\dots,\neg x_B\}$, with repetitions allowed; i.e.\ for every $\tilde{b}\in [B^{\uparrow}]_+$ we have
    \[
        x_{\tilde{b}} 
        \in 
        \{x_b\}_{b=1}^B
        \cup
        \{\neg x_b\}_{b=1}^B
    ,
    \]
    \item[(iv)] every gate in $\mathcal{C}$ belongs to $\{\operatorname{OR},\operatorname{AND},\operatorname{PROJ}_A\}=\{g_8,g_2,g_4\}$.
\end{enumerate}
The number of gates is exactly $B^{\uparrow}-1\le 2 S 2^S-1$
 and $B^{\uparrow}  \le  2 S 2^S$, and $
\Delta   \le  S+\big\lceil\log_2 S\big\rceil$.
\hfill\\
If, moreover, there exists some $C\in \mathbb{N}$ such that $S\le C \lceil \log(B)\rceil$ such that
\begin{equation}
\label{eq:SJunta__SparistyCondition}
S\le C \lceil \log(B)\rceil
\end{equation}
then there exists some $p,c\in \mathbb{N}_+$ independent of $B$ such that $B^{\uparrow}\le c\, B^p\,\lceil \log_2(B)\rceil$.
\end{lemma}

\subsubsection{Proofs of Results from Subsection~\ref{s:Proof__ss:TreeConstruction}}
\label{s:Proofs__ss:TreeConstruction___sss:Details}
\begin{proof}[{Proof of Lemma~\ref{lem:complete_binary_tree_representation}}]
If $f=  0$ is the constant false function, then we may take the formula $x_1\operatorname{AND} \neg x_1$ and skip directly to the padding step below, with
$N_{\mathrm{lit}}=2\le B 2^B$; cf.~\eqref{eq:N_padding}.
Whence, we henceforth assume that $f\not=  0$.

\noindent 
\textbf{Step 1: DNF with a bounded number of literal occurrences.}
Write $f$ in canonical disjunctive normal form (CDN).
For each $a=(a_1,\dots,a_B)\in\{0,1\}^B$ with $f(a)=1$, define $T_a:\{0,1\}^B\to \{0,1\}$ for every $x\in \{0,1\}^B$ by
\begin{equation}
\label{eq:literally_this}
        T_a(x)
    \eqdef
        \bigwedge_{i=1}^B \ell_i^{(a)}(x),
\mbox{ where }
    \ell_i^{(a)}(x)
    \eqdef
    \begin{cases}
        x_i, & a_i=1,\\[2pt]
        \neg x_i, & a_i=0
    \end{cases}
\end{equation}
for every $i\in [B]_+$.
Then, for every $x\in \{0,1\}^B$, 
$
    f(x)
    =
    \bigvee_{\{a: f(a)=1\}} T_a(x)
$.
Let $T$ be the number of satisfying assignments of $f$.
Then $1\le T\le 2^B$, and each $T_a$ uses exactly $B$ literals, so the total number of literal occurrences
$N_{\mathrm{lit}}$ satisfies
\begin{equation}
\label{eq:N_padding}
        N_{\mathrm{lit}}
    \eqdef
        T B
    \le
        B 2^B
.
\end{equation}
Thus we obtain a formula over $\{\operatorname{AND},\operatorname{OR}\}$ with at most $N_{\mathrm{lit}}$ literal leaves.

\noindent
\textbf{Step 2: Conversion into a Binary Tree.}
\hfill\\
\textbf{Remark:} \textit{We now replace the formula from Step $1$ by an equivalent binary formula whose depth we can bound explicitly.}

\medskip
\noindent
Fix $a$ with $f(a)=1$.
We claim that the conjunction $
    T_a(x)
    =
    \bigwedge_{i=1}^B \ell_i^{(a)}(x)
$ can be computed by a circuit whose underlying DAG is a binary tree of $\operatorname{AND}$-gates of depth at most $\lceil \log_2 B\rceil$, with leaves given by the literals $\ell_1^{(a)}(x),\dots,\ell_B^{(a)}(x)$.  To see this, we define the  circuit for $T_a$ level-by-level:
For the $0^{th}$ level, the nodes are the $B$ input literals $z^{(0)}_1,\dots,z^{(0)}_B$, where for every $i\in [B]_+$, we have set $z^{(0)}_i \eqdef \ell_i^{(a)}(x)$.

\noindent
We proceed iteratively for level $k\in \mathbb{N}_+$ large enough.  Suppose that level $k-1$ has $m_{k-1}$ nodes and denote 
$
    z^{(k-1)}_1,\dots,z^{(k-1)}_{m_{k-1}}
$.  Now, for each $j\in [ \lfloor m_{k-1}/2\rfloor]_+$, we introduce an $\operatorname{AND}$-gate with inputs $z^{(k-1)}_{2j-1}$ and $z^{(k-1)}_{2j}$, and denote its output by $
z^{(k)}_j \eqdef z^{(k-1)}_{2j-1} \land z^{(k-1)}_{2j}
$.  If $m_{k-1}$ is even, we are done, otherwise there is exactly one remaining node $z^{(k-1)}_{m_{k-1}}$ which has not been paired. In this case, set 
$
    z^{(k)}_{m_k} \eqdef z^{(k-1)}_{m_{k-1}}
$; where $m_k \eqdef \lfloor m_{k-1}/2\rfloor+1$. 
Thus, 
\begin{equation}
\label{eq:mk_bound}
    m_k 
\le 
    \Big\lceil \frac{m_{k-1}}{2}\Big\rceil
.
\end{equation}
Since $m_0 = B$, our recursion yielded
\begin{equation}
\label{eq:mk_bound__2MoreBounds_on_stuff}
    m_k 
\le 
    \Big\lceil \frac{B}{2^k}\Big\rceil
.
\end{equation}
Therefore, for $K\eqdef \lceil \log_2 B\rceil$ we have 
$
        m_K 
    \le 
        \Big\lceil \frac{B}{2^{\lceil\log_2 B\rceil}}\Big\rceil 
    = 
        1
$. 
In particular, after at most $K$ levels, there is exactly one remaining node $z^{(K)}_1$.

By construction, every $z^{(k)}_j$ is the conjunction of a subset of the literals $\{\ell_i^{(a)}(x)\}_{i=1}^B$, and these subsets form a partition of $\{1,\dots,B\}$ across all nodes at level $k$. Hence the final node $z^{(K)}_1$ is equal to
\begin{equation}
\label{eq:THEFINAL_NOODLE}
        z^{(K)}_1
    =
        \bigwedge_{i=1}^B \ell_i^{(a)}(x)
    =
        T_a(x)
.
\end{equation}
Moreover, every internal node in this construction has fan-in $2$, so the underlying DAG is a binary tree whose depth is at most $K=\lceil\log_2 B\rceil$.

\medskip
\noindent
Alternatively, if we have a disjunction of the form
\[
    f(x)
    =
    \bigvee_{\{a:\,f(a)=1\}} T_a(x).
\]
There are $T$ such terms. Applying exactly the same construction as above but replacing $\operatorname{AND}$ by $\operatorname{OR}$ and the inputs $z^{(0)}_j$ by the $T_a(x)$, we obtain a binary tree of $\operatorname{OR}$-gates which computes $\bigvee_{\{a:\,f(a)=1\}} T_a(x)$ and has depth at most $\lceil\log_2 T\rceil$.

\medskip
\noindent
Putting it all together; upon composing the $\operatorname{AND}$-trees for each $T_a$ with the top-level $\operatorname{OR}$-tree, we obtain a binary formula tree $T_0$ computing $f$ such that:
\begin{enumerate}[leftmargin=2em]
    \item[(a)] the leaves of $T_0$ are exactly the literal occurrences $\ell_i^{(a)}(x)$ from the DNF in Step $1$, so their total number is 
    \begin{equation}
    \label{eq:Bound_on_literals}
            N_{\mathrm{lit}}
        \le 
            B\,2^B
        ,
    \end{equation}
    \item[(b)] Since $T\le 2^B$ then 
    the depth of $T_0$ is at most
\begin{equation}
\label{eq:Bound_on_depth}
        \lceil\log_2 B\rceil + \lceil\log_2 T\rceil
        \le
        \lceil\log_2 B\rceil + B
    .
\end{equation}
\end{enumerate}
This $T_0$ is the binary tree we will pad in the next step.

\noindent
\textbf{Step 3: Conversion to a complete binary tree via projection gates}
\begin{equation}
\label{eq:who_is_deptho}
        \Delta
    \eqdef
        B+\big\lceil\log_2 B\big\rceil
\mbox{ and }
        B^{\uparrow}
    \eqdef
        2^\Delta
.
\end{equation}
Thus, the depth of $T_0$, denoted by $\operatorname{depth}(T_0)$ satisfies $\operatorname{depth}(T_0)\le \Delta$ and every leaf of $T_0$ lies at some depth $\ell\le \Delta$.
For any leaf $v$ of $T_0$ lying a level/layer of depth $\ell<\Delta$ with $\ell \in \mathbb{N}$.  We iteratively replace $v$ by a new internal gate $u$ decorated as $\operatorname{PROJ}_A$, then, we may make $v$ the left child of $u$, and add a new leaf $w$ as the right child of $u$, labelled by any literal in $\{x_1,\dots,x_B,\neg x_1,\dots,\neg x_B\}$. Then the value at $u$ equals the value at $v$, since $\operatorname{PROJ}_A(x,y)=x$ ignores its second input.
This increases the depth of that branch by $1$ and increases the total number of leaves by $1$.

Iterating this operation on each leaf with depth $<\Delta$ until all leaves lie at depth $\Delta$ yields a new tree $T_1$ such that: 1) every internal node has two children (the original $\operatorname{AND},\operatorname{OR}$ gates and the newly added $\operatorname{PROJ}_A$ gates), and 2) every leaf lies at depth exactly $\Delta$.  
Hence $T_1$ is a complete binary tree of depth $\Delta$.
Any finite full, i.e.\ every internal node has out-degree $2$, binary tree of depth $\Delta$ has exactly $2^\Delta$ leaves.  Consequently, 
$
    B^{\uparrow}
=
    2^\Delta
$ and the root of $T_1$ still computes $f$, all leaves are labelled by literals as in (ii), and all gates belong to $\{\operatorname{OR},\operatorname{AND},\operatorname{PROJ}_A\}$ as in (iii).
Finally, by the choice of $\Delta$,
\[
        B^{\uparrow}
    =
        2^\Delta
    =
        2^{B+\lceil\log_2 B\rceil}
    =
        2^B 
        \,
        2^{\lceil\log_2 B\rceil}
    \le
        2^B (2B)
    =
       2 B 2^B
.
\]
Now, since $\log_2(B 2^B)=B+\log_2 B$, then we also find that
$
    \Delta
=
    B+\big\lceil\log_2 B\big\rceil
\le
    \big\lceil\log_2(B 2^B)\big\rceil
$.

\nonumber
\textbf{Step 4: Final Counts.}
It remains to count the number of gates.  
Let $I$ denote the number of internal nodes of $T_1$, i.e.\ those nodes decorated by gates.
Counting edges in two ways gives
\[
        2I
    =
        \#\{\text{edges}\}
    =
        I + B^{\uparrow} - 1
.
\]
Consequently, $B^{\uparrow} = I+1$.
Whence, $I=B^{\uparrow}-1\le 2 B 2^B-1$, which is the number of gates of $\mathcal{C}$.
\end{proof}

\begin{proof}[{Proof of Lemma~\ref{lem:sJunta_as_embeddedpermutation}}]
If $f$ is an $S$-junta, there is a set $J\subseteq[B]_+$ of relevant variables with $|J|\le S$.  
Enlarge $J$ to $\{i_1,\dots,i_S\}$ if needed.  
Define $f_S:\{0,1\}^S\to\{0,1\}$ by
$f_S(z_1,\dots,z_S)\eqdef f(x)$ for any $x$ with $x_{i_j}=z_j$; this is well-defined
because $f$ does not depend on coordinates outside $\{i_1,\dots,i_S\}$.  
Then $f(x)=f_S(x_{i_1},\dots,x_{i_S})$ for all $x$.
Conversely, if $f(x)=f_S(x_{i_1},\dots,x_{i_S})$, then changing any coordinate
outside $\{i_1,\dots,i_S\}$ leaves $f(x)$ unchanged, so $f$ has at most $S$ relevant variables; whence, $f$ is an $S$-junta.
\end{proof}

\begin{proof}[{Proof of Lemma~\ref{lem:JuntaTreeification}}]
All but the last claims are a direct consequence of Lemmata~\ref{lem:complete_binary_tree_representation} and~\ref{lem:sJunta_as_embeddedpermutation}.
For the final claim, note that if there exists some $C\in \mathbb{N}$ such that $S\le C \lceil \log(B)\rceil$ then~\eqref{eq:boundsboundybounds} implies that
\[ 
    B^{\uparrow}
\le 
    S2^S 
\le 
    C\lceil\log_2 B\rceil\,2^{C\lceil\log_2 B\rceil} 
\le 
    C\,2^C\,B^C\,\lceil\log_2 B\rceil
.
\] 
Setting $p\eqdef C$ and $c\eqdef C2^C$ yields the result.
\end{proof}

We are now in place to prove our universal approximation guarantee.  We prove Theorem~\ref{thrm:UniversalReasoning__Quantitative} as it implies~\ref{thrm:UniversalReasoning} upon taking $S=B$.
\begin{proof}[{Proof of Theorem~\ref{thrm:UniversalReasoning}}]
Fix our failure probability $\delta\in (0,1]$.
\hfill\\
\textbf{Step $0$ - Setup:}
\hfill\\
Fix a temperature hyperparameter $\eta>0$, to be determined retroactively.  
Since $f:\{0,1\}^B\to \{0,1\}$ is an $S$-Juntas then applying Lemma~\ref{lem:JuntaTreeification}, we may represent $f$ as in~\eqref{eq:dim_reduced_compute}.
Now the computation graph $G=(V,E)$ of the circuit $\mathcal{C}$ computing $f_S$, in the notation of Lemma~\ref{lem:JuntaTreeification}, is a rooted complete binary tree of depth $\Delta\in \mathbb{N}_+$.  Therefore we may decompose $V$ as
\begin{equation}
\label{eq:partitioning_V}
    V
=
    \bigsqcup_{l=0}^\Delta
    \,
        V_l
\mbox{ where }
    V_l = \{v_{l:i}\}_{i=1}^{2^{\Delta-l}}
\end{equation}
and we may partition $E$ as
\begin{equation}
\label{eq:partitioning_E}
    E
=
    \bigsqcup_{l=1}^\Delta
    \,
        E_l
\mbox{ where }
    E_l 
= 
    \bigsqcup_{i=1}^{2^{\Delta-l}}
    \,
        \big\{
            (v_{l-1:i_1},v_{l:i})
            ,
            (v_{l-1:i_2},v_{l:i})
        \big\}_{i=1}^{2^{\Delta-l}}
\end{equation}
where, for every $i\in [2^{\Delta-l}]_+$ the indices $i_1,i_2\in [2^{\Delta-l+1}]_+$ in~\eqref{eq:partitioning_E} are distinct and unique.  Moreover, WLOG, we may arrange $v_{l:i}$ the edges $E$ in the graph $G$ do not cross, as trees are planar.

Furthermore, since the gates in $\mathcal{C}$ only decorate the vertices in $\sqcup_{l=1}^\Delta\, V_l$ and not in $V_0$ (those are only variables described by Lemma~\ref{lem:JuntaTreeification}~(iii)).  We denote the gate decorating the vertex $v_{l:i}$ by $g_{l:i}$ and, WLOG, we identify $g_{l:i} \in \{g_k\}_{k=1}^K$ (say $g_{l:i}=g_k$) with the unique standard basis vector in $\{e_k\}_{k=1}^{16}$ with $1$ only on the entry of the index $k$).

\hfill\\
\textbf{Step $1$ - Representation:}
\hfill\\
Consider an 
$\randcirc$
with representation~\eqref{eq:Model_Def} having depths $(2^\Delta,2^{\Delta-1},\dots,2^0)=(2^{\Delta-1})_{l=0}^\Delta\,\in \mathbb{N}_+^{\Delta+1}$; we determine its internal parameters now.
Now for every $l\in [\Delta]_+$ and each $i\in [2^{\Delta-l}]_+$ let
\begin{equation}
\label{eq:parameters}
    W^{(l:\sigma)}
\eqdef 
    \big(\eta\,g_{l:i}\big)_{i=1}^{2^{\Delta-l}}
,\,\,
    W^{(l:1)}
\eqdef 
    \eta
    \oplus_{i=1}^{2^{\Delta-l}}\,
        e_{i_1}
\mbox{,and}\,\,
    W^{(l:2)}
\eqdef 
    \eta
    \oplus_{i=1}^{2^{\Delta-l}}\,
        e_{i_2}
.
\end{equation}
Now applying Propositions~\ref{prop:random_weight_assignment} and~\ref{prop:Good_Properties_Uncrossable} and taking a union bound, the probability that $
\randcirc(x,\omega)=f(x)$ for every element of the Boolean cube $x\in \{0,1\}^{B^{\uparrow}}$ is at-least
\begin{equation}
\label{eq:Success_probability__fS}
1-
\sum_{l=1}^\Delta\,
\sum_{i=1}^{2^{\Delta-l}}\,
\tilde{\delta}(\eta)
\end{equation}
where $\tilde{\delta}(\eta)\in (0,1]$ depends continuously on $\eta>0$; with $\lim\limits_{\uparrow \infty}\, \tilde{\delta}(\eta)=0$.  
Now, applying Proposition~\ref{prop:lifting_channels}, we may choose $W\in \mathbb{R}^{B^{\uparrow}\times 2B}$, and thus, $X^{\uparrow:W}$ as as defined in~\eqref{eq:lifted_guy}, so that $X^{\uparrow:W}$ coincides with $\Pi$
, with probability at-least 
\begin{equation}
\label{eq:Success_probability__Pi}
1-\bar{\delta}(\eta)
\end{equation}
for some $\bar{\delta}(\eta)$ depending continuously on $\eta$ and such that $\lim\limits_{\eta \uparrow \infty}\bar{\delta}(\eta)=0$.  
Consequently, we may choose $\eta$ large enough so that 
\[
    \max\big\{
        \bar{\delta}(\eta)
        ,
        \tilde{\delta}(\eta)
    \big\}
\le 
    \delta\,
    \Big(
        \max\Big\{2,2^\Delta-1\Big\}
    \Big)^{-1}
\]
where we note that $\sum_{l=1}^\Delta\,
\sum_{i=1}^{2^{\Delta-l}}\,1=2^\Delta-1$.  Taking another union bound, our choice of $\eta$ and the lower-bounds in~\eqref{eq:Success_probability__fS} and~\eqref{eq:Success_probability__Pi} yield the conclusion.
\end{proof}

\section{Experimental Setup and Implementation Details}
\label{app:exp_impl}

This appendix documents the experimental protocol and implementation-facing details used to obtain the numerical results reported in the main text. The theoretical results in Sections~\ref{s:Model}--\ref{s:Theory} are agnostic to these choices; the goals here are (i) reproducibility and (ii) a precise description of the evaluation and interpretability probes used for continuous baselines.

\subsection{Task family and dataset generation}
\label{app:dataset}

Each instance is a Boolean function $f:\{0,1\}^B\to\{0,1\}$ specified by a syntactic formula (SOP-style with additional macros; see \texttt{src/ubcircuit/tasks.py}). For every instance we train and evaluate on the full truth table, i.e., all $2^B$ assignments. We report \emph{exact match} (EM), meaning the predicted truth table equals the target truth table on all $2^B$ inputs.

\paragraph{Dataset-provided structural proxies.}
Each dataset row provides generation-time proxies
\[
S_{\mathrm{base}}
\quad\text{and}\quad
L_{\mathrm{base}},
\]
where $S_{\mathrm{base}}$ is a width proxy used during generation (e.g.\ a term-count/arity proxy), and $L_{\mathrm{base}}$ is a heuristic depth estimate. These are treated as \emph{informative but not privileged}: the benchmark formulas are not algebraically simplified, and multiple syntactically distinct formulas may be logically equivalent.

\subsection{Models and matching regimes}
\label{app:models}

\paragraph{\ModelNameShort\ (our model).}
We use the 16-gate version throughout (Table~\ref{tab:all_possible_gates}), with fan-in-$2$ gates and probabilistic routing. After training, we decode a discrete circuit by taking layerwise $\argmax$ decisions over the learned categorical distributions, yielding (i) a decoded circuit expression and (ii) gate-usage histograms (Appendix~\ref{app:gate_hists}).

\paragraph{MLP baselines.}
We use a ReLU MLP trained on the same truth tables. We consider three matching regimes:
\begin{itemize}[leftmargin=1.5em,itemsep=0.25em]
\item \textbf{Neuron-match:} MLP hidden width is set to $S_{\mathrm{model}}$ and depth to $L_{\mathrm{model}}$ (matching the \ModelNameShort\ ``shape'').
\item \textbf{Param-match (soft):} choose MLP width so trainable parameter count is $\le$ \ModelNameShort \emph{trainable} parameter count for the same $(B,S_{\mathrm{model}},L_{\mathrm{model}})$.
\item \textbf{Param-match (total):} choose MLP width so trainable parameter count is $\le$ \ModelNameShort\ trainable parameters plus a fixed primitive-count proxy (treating each primitive-gate basis element as a fixed component).
\end{itemize}

\subsection{Width/depth scaling rules}
\label{app:scaling}

To study the effect of additional ``budget'' beyond dataset proxies, we apply scaling rules:
\[
S_{\mathrm{model}} = \mathrm{Clamp}\big(\mathrm{Op}_S(S_{\mathrm{base}};k_S)\big),
\qquad
L_{\mathrm{model}} = \mathrm{Clamp}\big(\mathrm{Op}_L(L_{\mathrm{base}};k_L)\big),
\]
where $\mathrm{Op}$ is either \texttt{identity}, \texttt{add}, or \texttt{mul}, and \texttt{Clamp} enforces preset bounds (e.g.\ $S_{\min}\le S_{\mathrm{model}}\le S_{\max}$). In the main \ModelNameShort-vs-MLP comparison we use the additive setting
\[
S_{\mathrm{model}} = S_{\mathrm{base}} + S_{\mathrm{add}},
\qquad
L_{\mathrm{model}} = L_{\mathrm{base}} + L_{\mathrm{add}},
\]
with $S_{\mathrm{add}}=10$ and $L_{\mathrm{add}}=0$ unless stated otherwise. Sweep plots in Appendix~\ref{app:sweep} vary $S_{\mathrm{add}}$ and $L_{\mathrm{add}}$.

\subsection{Training system used in experiments}
\label{app:impl_heuristics}

This section describes the concrete training system used in our experiments. These choices are implementation-level stabilizers for optimizing a differentiable circuit model; they are not required for the proof of universality.

\subsubsection{DepthStack architecture.}
Our implementation (\texttt{src/ubcircuit/modules.py}) builds a depth-$L_{\mathrm{model}}$ stack:
(i) optional \texttt{BitLifting} mapping $B\mapsto B_{\mathrm{eff}}$,
(ii) a first \texttt{GeneralLayer} mapping $B_{\mathrm{eff}}\to 2$ wires (with pair selection when $B_{\mathrm{eff}}>2$),
(iii) $L_{\mathrm{model}}-2$ middle \texttt{ReasoningLayer}s mapping $2\to 2$ wires, and
(iv) a final \texttt{ReasoningLayer} mapping $2\to 1$ output.
Each layer has $S_{\mathrm{model}}$ parallel \texttt{BooleanUnit}s (``heads''), and a row-wise softmax mixer $L_\ell$ routes head outputs into the next-layer wires.

\subsubsection{Gate interpolants for the 16-gate head}
\label{app:impl_sigma16}

Our implementation uses a differentiable feature head $\sigma_{16}:[0,1]^2\to[0,1]^{16}$ that simultaneously evaluates the $16$ fan-in-$2$ Boolean gates in Table~\ref{tab:all_possible_gates}.  
For each gate $g$ with truth table vector $Z^{(g)}\in\{0,1\}^4$ ordered by corners
$(0,0),(0,1),(1,0),(1,1)$, we define an interpolant
\[
\sigma(x \mid Z^{(g)}) \;=\; \sum_{c\in\{00,01,10,11\}} Z^{(g)}_c\,\phi_c(x),
\qquad
\sigma_{16}(x)\;=\;\big(\sigma(x\mid Z^{(g)})\big)_{g=1}^{16},
\]
where $\{\phi_c\}$ is a (mode-dependent) set of corner basis functions satisfying
$\phi_c(\bar x_{c'})=\delta_{c,c'}$ at the four Boolean corners $\bar x_{00},\bar x_{01},\bar x_{10},\bar x_{11}$.

\paragraph{Mode \texttt{lagrange} (multilinear).}
Let $\bar x_{00}=(0,0)$, $\bar x_{01}=(0,1)$, $\bar x_{10}=(1,0)$, $\bar x_{11}=(1,1)$.  
Define the bilinear Lagrange basis on $[0,1]^2$:
\[
\phi_{00}(x)=(1-x_1)(1-x_2),\quad
\phi_{01}(x)=(1-x_1)x_2,\quad
\phi_{10}(x)=x_1(1-x_2),\quad
\phi_{11}(x)=x_1x_2.
\]
This basis satisfies $\phi_{ij}(\bar x_{kl})=\delta_{(ij),(kl)}$ and yields an exact multilinear extension of each gate.

\paragraph{Mode \texttt{rbf} (normalized Gaussian kernels).}
Let $w_c(x)=\exp(-\|x-\bar x_c\|_2^2/(2s^2))$ for bandwidth $s>0$ and corners $\bar x_c$.  
Define a partition-of-unity basis by normalization:
\[
\phi_c(x)\;=\;\frac{w_c(x)}{\sum_{c'} w_{c'}(x)}.
\]
This yields a smooth ($C^\infty$) interpolant whose sharpness is controlled by $s$.

\paragraph{Mode \texttt{bump} (compactly supported $C^\infty$ kernels).}
Let $w_c(x)=\psi(\|x-\bar x_c\|_2/r)$ where $r>0$ is a radius and $\psi$ is a standard smooth compactly supported bump with $\psi(t)=0$ for $t\ge 1$.  
We again normalize:
\[
\phi_c(x)\;=\;\frac{w_c(x)}{\sum_{c'} w_{c'}(x)}.
\]
In our code, \texttt{sigma16.radius} controls $r$.  When the support radius is chosen so that at each corner $\bar x_c$ all other kernels vanish
(e.g., $r < 1$, since corner-to-corner distances are $\ge 1$), the normalized basis satisfies
$\phi_c(\bar x_{c'})=\delta_{c,c'}$.

\subsubsection{Layerwise bandwidth schedule.}
For \texttt{rbf} and the internal bandwidth parameter used by \texttt{sigma16}, we optionally anneal the sharpness across layers using a linear schedule
\[
s_\ell \;=\; s_{\mathrm{start}} + \frac{\ell}{L_{\mathrm{model}}-1}\,(s_{\mathrm{end}}-s_{\mathrm{start}}),
\qquad \ell=0,\dots,L_{\mathrm{model}}-1.
\]
matching the implementation parameters \texttt{sigma16.s\_start} and \texttt{sigma16.s\_end}.  
We report an ablation over \texttt{sigma16.mode} in Table~\ref{tab:sigma16_mode_ablation}.

\subsubsection{Pair selection and MI-based priors.}
\label{app:impl_pair}
For $B_{\mathrm{eff}}>2$, the first layer uses a differentiable pair selector with two row-softmax matrices $\mathrm{PL},\mathrm{PR}\in\mathbb{R}^{S_{\mathrm{model}}\times B_{\mathrm{eff}}}$ that define categorical distributions over left/right inputs.
When \texttt{pair.route=mi\_soft}, we compute instance-specific high-mutual-information pairs from the truth table (via \texttt{top\_mi\_pairs} in \texttt{src/ubcircuit/mi.py}) and convert them to priors $\mathrm{PL}_{\mathrm{prior}},\mathrm{PR}_{\mathrm{prior}}\in \Delta_{B_{\mathrm{eff}}}^{S_{\mathrm{model}}\times B_{\mathrm{eff}}}$ (expanded to $B_{\mathrm{eff}}$ if lifting is enabled). These are injected as log-probability biases with strength \texttt{pair.prior\_strength}.
When \texttt{pair.route=mi\_hard}, the MI-selected pairs are used as fixed discrete pairs (\texttt{FixedPairSelector}).
When \texttt{pair.route\allowbreak= learned}, no priors are used.

\subsubsection{Repulsive right-pick (optional).}
To discourage self-pairs within a head, the pair selector optionally couples the right-pick to the left-pick.
If \texttt{pair.repel=true}, the implementation supports:
(i) log-space repulsion (\texttt{pair.mode=log} or \texttt{hard-log}), adding $\eta \log(1-\mathrm{PL})$ to right logits and optionally masking the left argmax coordinate in \texttt{hard-log}; and
(ii) probability-space repulsion (\texttt{pair.mode=mul} or \texttt{hard-mul}), multiplying right probabilities by $(1-\mathrm{PL})$ and optionally hard-zeroing the left argmax coordinate in \texttt{hard-mul}.
Our best-reported configuration uses MI priors (\texttt{pair.route=mi\_soft}) with repulsion disabled.

\subsubsection{Asynchronous temperature annealing.}
Each layer uses a temperature $\tau_\ell(t)$ to control the sharpness of both gate selection and routing.
We use asynchronous schedules with
\texttt{direction} $\in \{\texttt{top\_down}, \texttt{bottom\_up}\}$ and
\mbox{\texttt{schedule}} $\in \{\mbox{\texttt{linear}}, \mbox{\texttt{cosine}}\}$. Early in training, larger $\tau_\ell$ keeps selections soft (smoother gradients); later, $\tau_\ell\downarrow T_{\min}$ makes distributions near one-hot, enabling stable $\argmax$ decoding.

\subsubsection{Sigma16 bandwidth schedule.}
When using the 16-gate head, we optionally anneal the gate-feature bandwidth across layers (\texttt{sigma16.s\_start} $\to$ \texttt{sigma16.s\_end}) by calling \texttt{DepthStack.set\_layer\_taus\_and\_bandwidths}. This is a per-layer bandwidth on the interpolant (not the softmax temperature).

\subsubsection{Objective, optimizer, and early stopping.}
\label{app:impl_training_aids}
All methods are trained on full truth tables using binary cross entropy. \ModelNameShort\ optimizes
\[
\mathcal{J}
=
\mathcal{L}_{\mathrm{BCE}}
+
\mathcal{R}_{\mathrm{ent}}
+
\mathcal{R}_{\mathrm{div\text{-}units}}
+
\mathcal{R}_{\mathrm{div\text{-}rows}}
+
\mathcal{R}_{\mathrm{const16}}.
\]
We use RMSProp in the main runs. Early stopping monitors EM after a warm-up (configurable \texttt{min\_steps}, \texttt{check\_every}, \texttt{patience\_checks}).

\subsubsection{Regularizers \texorpdfstring{$\mathcal{R}_{\bullet}$ }{}used in experiments.}
In the implementation (\texttt{src/ubcircuit/utils.py}, \texttt{regularizers\_bundle}), the regularizers are computed from the per-layer diagnostics returned by the forward pass. 
For each layer $\ell$, the model exposes:
(i) routing rows $L_\ell\in\Delta_{S}^{B_{\ell+1}\times S}$ (row-softmax over heads),
and (ii) head gate logits $\{W^{(\ell:s)}\}_{s=1}^{S}$ (each $W^{(\ell:s)}\in\mathbb{R}^{16}$ when \texttt{gate\_set=16}).
Let $\tau_\ell$ be the layer temperature used at the current optimization step.

\emph{Entropy term.}
The code sums Shannon entropies of each routing row and each head's gate distribution:
\[
\mathcal{R}_{\mathrm{ent}}
=
\lambda_{\mathrm{ent}}
\sum_{\ell}\Bigg(
\sum_{k=1}^{B_{\ell+1}} H\!\left(L_\ell[k,:]\right)
+
\sum_{s=1}^{S} H\!\left(\operatorname{SM}(W^{(\ell:s)}/\tau_\ell)\right)
\Bigg),
\qquad
H(q)\triangleq-\sum_i q_i\log q_i .
\]

\emph{Head-diversity term.}
Let $p^{(\ell:s)}=\operatorname{SM}(W^{(\ell:s)}/\tau_\ell)\in\Delta_{16}$. The head-diversity penalty is the sum of pairwise cosine similarities:
\[
\mathcal{R}_{\mathrm{div\text{-}units}}
=
\lambda_{\mathrm{div\text{-}units}}
\sum_{\ell}\sum_{1\le s<s'\le S}
\cos\!\Big(p^{(\ell:s)},p^{(\ell:s')}\Big),
\qquad
\cos(u,v)=\frac{\langle u,v\rangle}{\|u\|_2\|v\|_2}.
\]

\emph{Row-diversity term.}
The routing diversity penalty is the sum of pairwise cosine similarities between routing rows:
\[
\mathcal{R}_{\mathrm{div\text{-}rows}}
=
\lambda_{\mathrm{div\text{-}rows}}
\sum_{\ell}\sum_{1\le k<k'\le B_{\ell+1}}
\cos\!\Big(L_\ell[k,:],L_\ell[k',:]\Big).
\]

\emph{Constant-gate penalty (16-gate head).}
When enabled (\texttt{lam\_const16}$>0$), the code adds a penalty on selecting \texttt{FALSE} and \texttt{TRUE} on all non-final layers:
\[
\mathcal{R}_{\mathrm{const16}}
=
\lambda_{\mathrm{const16}}
\sum_{\ell<L_{\mathrm{model}}-1}\sum_{s=1}^{S}
\Big(p^{(\ell:s)}_{\texttt{FALSE}}+p^{(\ell:s)}_{\texttt{TRUE}}\Big),
\]
and uses $\mathcal{R}_{\mathrm{const16}}=0$ on the final layer. (In the best-reported setting we set $\lambda_{\mathrm{const16}}=0$.)

\subsubsection{Decoding and evaluation.}
\label{app:impl_decoding}
Decoding is post-hoc: after training we select gates and routing via $\argmax$ at each layer to obtain a discrete circuit and a symbolic expression. EM is computed from the model's truth-table predictions (threshold at $0.5$) and is independent of symbolic string equality.

\subsection{Proof-level construction vs.\ experimental instantiation}
\label{app:proof_vs_impl}

The theory in Sections~\ref{s:Model}--\ref{s:Theory} studies a randomized circuit-valued model and proves universal representability and circuit-structured realizations. The experimental system instantiates the same circuit template with a concrete differentiable parameterization and a deterministic decoding protocol. The main differences that matter for interpreting the numerics are summarized in Table~\ref{tab:proof_vs_impl}. In particular, several training-time heuristics are used in practice (e.g.\ temperature schedules and regularization) to stabilize optimization and promote crisp $\argmax$ decoding; these are not required by the existence proofs.

\begin{table}[t]
\centering
\small
\setlength{\tabcolsep}{6pt}
\renewcommand{\arraystretch}{1.15}
\begin{adjustbox}{max width=\linewidth}
\begin{tabular}{p{0.18\linewidth} p{0.34\linewidth} p{0.44\linewidth}}
\toprule
Component & Proof-level statement & Experimental instantiation (with refs.) \\
\midrule
Input lifting &
Stochastic bit-lifting channel $\mathbf{X}^{\uparrow:W}$ is included (Section~\ref{s:LiftingChannels}). &
Implemented (\texttt{BitLifting}), but best-performing setting disables lifting (\texttt{lifting.use=false}). \\

Gate family &
Randomized selection over all $16$ fan-in-$2$ gates (\cref{tab:all_possible_gates}). &
Uses the 16-gate head (\texttt{gate\_set=16}) with smooth $\sigma_{16}$ interpolation during training and $\argmax$ decoding at evaluation (\cref{app:impl_sigma16,app:impl_decoding}). \\

Edge selection &
Structured selection discouraging degenerate pairing (\cref{s:Model__ss:UncrossableEdges}). &
First layer uses \texttt{PairSelector} with learned picks or MI-driven priors (\texttt{mi\_soft}/\texttt{mi\_hard}) and optional repulsion modes (\texttt{log}/\texttt{mul}, \texttt{hard-*}) (\cref{app:impl_pair}). \\

Decoding &
Circuit realization is defined per outcome $\omega$ (\cref{thrm:StructureAware}). &
Post-hoc $\argmax$ decoding yields a single discrete circuit expression and gate histograms (\cref{app:impl_decoding,app:gate_hists}). \\

Optimization aids &
Not required for universality/existence proofs. &
Async temperature annealing and regularizers used for stable optimization and crisp decoding (\cref{app:impl_training_aids}). \\
\bottomrule
\end{tabular}
\end{adjustbox}
\caption{High-level differences between the proof-level construction and the training system used in experiments.}
\label{tab:proof_vs_impl}
\end{table}

\subsection{Interpretability metrics for continuous baselines}
\label{app:interpretability}

BNR (Appendix~\ref{app:BNR_detail}) is defined for Boolean-valued maps. MLP hidden units are real-valued, so we compute operational BNR diagnostics and additional gate-recoverability probes by enumerating each unit on the full truth table.

\paragraph{BNR diagnostics (rounded and tolerant).}
Let $u:\{0,1\}^B\to\mathbb{R}$ be a hidden unit evaluated on all $2^B$ inputs.
We report two operational variants:
\begin{itemize}[leftmargin=1.5em,itemsep=0.25em]
\item \textbf{BNR$_{\mathrm{exact}}$ (rounded).} Fix a rounding precision $d$ (we use $d=6$). Declare $u$ BNR if
\[
\#\big\{\mathrm{Round}_d(u(x)):\, x\in\{0,1\}^B\big\}\le 2.
\]
\item \textbf{BNR$_{\varepsilon}$ (two-cluster tolerance).} Fix $\varepsilon>0$ (we use $\varepsilon=10^{-3}$). Split the multiset $\{u(x)\}$ at its median, take the median of each half as centers $c_0,c_1$, and declare $u$ $\varepsilon$-BNR if
\[
\max_{x\in\{0,1\}^B}\min\big\{|u(x)-c_0|,\,|u(x)-c_1|\big\}\le \varepsilon.
\]
\end{itemize}
For a layer with $H$ units we report the fraction of units satisfying each criterion. For an $L$-layer MLP we report both the first-layer value (L1) and the average across hidden layers (denoted ``all''). \ModelNameShort\ is Boolean by construction, hence achieves BNR$=1$ for all internal units.

\paragraph{Deterministic binarization for gate probes.}
To probe Boolean decision patterns of MLP units, we map a real unit to a Boolean function by thresholding at the all-zero assignment:
\[
t(u)\triangleq u(\mathbf{0}),
\qquad
\hat u(x)\triangleq \mathbb{I}\{u(x)\ge t(u)\}\in\{0,1\}.
\]
This binarization is used only for the gate-recoverability probes below.

\paragraph{Primitive recoverability on the first hidden layer (input primitives).}
For each first-layer binarized unit $\hat u$, define
\[
\mathrm{bestAcc}_{\mathrm{in}}(u)
\triangleq
\max_{p\in\mathcal{P}_{\mathrm{in}}}
\Pr_{x\sim\mathrm{Unif}(\{0,1\}^B)}\big[\hat u(x)=p(x)\big],
\]
where $\mathcal{P}_{\mathrm{in}}$ contains literals $\{a_i,\neg a_i\}_{i=1}^B$ and the 16 fan-in-$2$ gates in Table~\ref{tab:all_possible_gates} applied to ordered input pairs $(a_i,a_j)$ (with $i\neq j$). We report
\[
\mathrm{primHit}_{\mathrm{in}}
\triangleq
\Pr_u[\mathrm{bestAcc}_{\mathrm{in}}(u)=1],
\qquad
\mathrm{primBest}_{\mathrm{in}}
\triangleq
\mathbb{E}_u[\mathrm{bestAcc}_{\mathrm{in}}(u)].
\]
Here \textbf{primHit} is an exact recovery rate (no tolerance), while \textbf{primBest} averages the best-match accuracies.

\paragraph{Layerwise primitive recoverability (all layers).}
Let $\hat z^{(\ell-1)}_1,\dots,\hat z^{(\ell-1)}_{H_{\ell-1}}$ denote the binarized units in layer $\ell-1$. For a binarized unit $\hat u^{(\ell)}_j$ in layer $\ell$, define
\[
\mathrm{bestAcc}^{(\ell)}_{\mathrm{layer}}(u_j)
\triangleq
\max_{\substack{i\neq k\\ g\in\mathbb{G}^{\star}}}
\Pr_{x\sim\mathrm{Unif}(\{0,1\}^B)}
\Big[
\hat u^{(\ell)}_j(x)=g(\hat z^{(\ell-1)}_{i}(x),\hat z^{(\ell-1)}_{k}(x))
\Big].
\]
We define per-layer summaries
\[
\mathrm{primHit}^{(\ell)}_{\mathrm{layer}}
\triangleq
\Pr_j[\mathrm{bestAcc}^{(\ell)}_{\mathrm{layer}}(u_j)=1],
\qquad
\mathrm{primBest}^{(\ell)}_{\mathrm{layer}}
\triangleq
\mathbb{E}_j[\mathrm{bestAcc}^{(\ell)}_{\mathrm{layer}}(u_j)],
\]
and report their averages across hidden layers (denoted ``all'').

\subsection{Gate histogram diagnostics}
\label{app:gate_hists}

We record gate-type histograms as descriptive diagnostics:
\begin{itemize}[leftmargin=1.5em,itemsep=0.25em]
\item \textbf{\ModelNameShort-all:} gate histogram across all units and layers under $\argmax$ decoding.
\item \textbf{\ModelNameShort-path:} gate histogram restricted to the decoded readout path.
\item \textbf{MLP-all:} histogram of exact gate matches from the layerwise recoverability probe, aggregated across layers.
\end{itemize}
These histograms are not correctness certificates; they help interpret which primitives dominate the decoded computation. (In our runs, MLP exact-match recoveries are often dominated by constant-type behaviors such as \texttt{TRUE}, so we interpret high \textbf{primBest} jointly with \textbf{primHit} and the gate histograms.)

\section{Additional Results and Diagnostics}
\label{app:results_full}

\subsection{EM and BNR}
\label{app:main_tables}

We summarize (i) task correctness via exact match (EM) and (ii) neuronwise Boolean-valuedness via BNR.
EM is computed on the full truth table for each instance (Section~\ref{app:dataset}); BNR is defined in Appendix~\ref{app:interpretability}.
Table~\ref{tab:em_short} reports EM for MLP baselines under each matching regime, while reporting the \ModelNameShort\ EM once under the fixed canonical \ModelNameShort\ setting used throughout.
Table~\ref{tab:bnr_only} reports BNR diagnostics: for \ModelNameShort\, BNR$\equiv 1$ holds by construction (Theorem~\ref{thrm:StructureAware}); for MLPs, BNR is measured post-hoc on hidden activations.

\begin{table}[t]
\centering
\small
\begin{tabular}{l c c c}
\toprule
MLP match & Avg.\ MLP params & MLP EM & \ModelNameShort\ EM \\
\midrule
neuron & 368 & 0.987 $\pm$ 0.002 & \multirow{3}{*}{0.971 $\pm$ 0.003} \\
param\_soft & 683 & 0.992 $\pm$ 0.001 & \\
param\_total & 1154 & 0.999 $\pm$ 0.002 & \\
\bottomrule
\end{tabular}
\caption{Exact-match (EM) on truth-table evaluation (mean $\pm$ std across 5 seeds). \ModelNameShort\ EM is reported once (canonical \ModelNameShort\ setting) and MLP EM is shown for each matching regime.}
\label{tab:em_short}
\end{table}

\begin{table}[t]
\centering
\small
\begin{tabular}{lcccc}
\toprule
Model / match
& BNR$_{\mathrm{exact}}$(L1)
& BNR$_{\mathrm{exact}}$(all)
& BNR$_{\varepsilon}$(L1)
& BNR$_{\varepsilon}$(all)
\\
\midrule
\ModelNameShort\ & 1.000 & 1.000 & 1.000 & 1.000 \\
\midrule
neuron & 0.114 $\pm$ 0.003 & 0.209 $\pm$ 0.002 & 0.115 $\pm$ 0.004 & 0.209 $\pm$ 0.002 \\
param\_soft & 0.107 $\pm$ 0.002 & 0.199 $\pm$ 0.001 & 0.107 $\pm$ 0.002 & 0.199 $\pm$ 0.001 \\
param\_total & 0.105 $\pm$ 0.002 & 0.195 $\pm$ 0.002 & 0.105 $\pm$ 0.002 & 0.196 $\pm$ 0.002 \\
\bottomrule
\end{tabular}
\caption{BNR diagnostics (mean $\pm$ std across 5 seeds). ``all'' denotes the average across hidden layers. \ModelNameShort\ is BNR by construction, hence BNR$\equiv 1$.}
\label{tab:bnr_only}
\end{table}

\noindent\textbf{Takeaway.}
Across matching regimes, MLP achieves higher EM than \ModelNameShort\ (Table~\ref{tab:em_short}), while \ModelNameShort\ guarantees neuronwise Boolean-valuedness (Table~\ref{tab:bnr_only}).
For MLPs, strict two-valuedness is rare, especially at early layers; the tolerance-based BNR$_\varepsilon$ closely tracks the rounded BNR$_{\mathrm{exact}}$ in this regime.

\subsection{Primitive recoverability}
\label{app:prim_table}

BNR measures whether hidden activations are strictly two-valued; it does not ask whether a unit behaves like a \emph{specific} Boolean gate.
To probe gate-like \emph{decision patterns} in MLPs, we report fan-in-$2$ primitive recoverability after deterministic binarization (Appendix~\ref{app:interpretability}).
Table~\ref{tab:prim_only} summarizes recoverability both (i) against primitives over \emph{input bits} (L1) and (ii) layerwise relative to binarized previous-layer units (all).

\begin{table}[t]
\centering
\small
\begin{tabular}{l c c c c}
\toprule
MLP match
& primHit$_{\mathrm{in}}$(L1)
& primBest$_{\mathrm{in}}$(L1)
& primHit$_{\mathrm{layer}}$(all)
& primBest$_{\mathrm{layer}}$(all)
\\
\midrule
neuron & 0.463 $\pm$ 0.004 & 0.945 $\pm$ 0.001 & 0.552 $\pm$ 0.002 & 0.953 $\pm$ 0.000 \\
param\_soft & 0.466 $\pm$ 0.003 & 0.941 $\pm$ 0.000 & 0.559 $\pm$ 0.002 & 0.953 $\pm$ 0.000 \\
param\_total & 0.477 $\pm$ 0.003 & 0.939 $\pm$ 0.001 & 0.570 $\pm$ 0.002 & 0.953 $\pm$ 0.000 \\
\bottomrule
\end{tabular}
\caption{Fan-in-2 primitive recoverability after deterministic binarization (mean $\pm$ std across 5 seeds). primHit is the fraction of units with an exact primitive match; primBest is the average best-match accuracy. ``in'' evaluates first-layer units against primitives over input bits; ``layer'' evaluates gate recoverability relative to binarized previous-layer units, averaged across hidden layers.}
\label{tab:prim_only}
\end{table}

\noindent\textbf{Takeaway.}
Primitive recoverability can be substantially higher than strict BNR (compare Tables~\ref{tab:bnr_only} and~\ref{tab:prim_only}), indicating that many MLP units induce gate-like \emph{Boolean patterns} after binarization even when their real-valued magnitudes remain multi-level.

\subsection{Gate histograms (\ModelNameShort\ vs.\ MLP)}
\label{app:gate_fig}

To contextualize Table~\ref{tab:prim_only}, we visualize which gate types dominate the recovered computations.
Figure~\ref{fig:gate_hist} compares: (i) \ModelNameShort-all (all decoded units), (ii) \ModelNameShort-path (decoded readout path), and (iii) MLP-all (exact recovered gates from the layerwise primitive probe, aggregated across layers) for each matching regime (Appendix~\ref{app:gate_hists}).

\begin{figure}[t]
\centering
\includegraphics[width=0.98\linewidth]{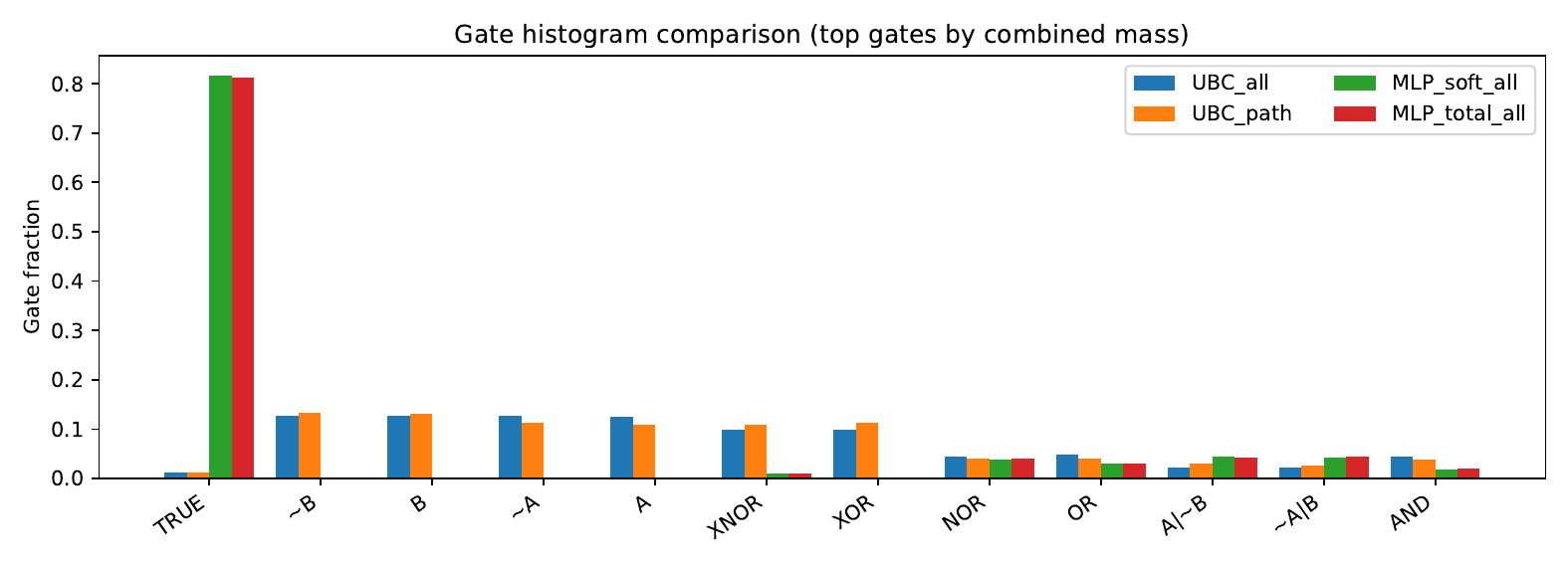}
\caption{Gate histograms: \ModelNameShort-all, \ModelNameShort-path, and MLP-all (exact recovered gates) across matching regimes.}
\label{fig:gate_hist}
\end{figure}

\noindent\textbf{Takeaway.}
Although Table~\ref{tab:prim_only} reports high \texttt{primBest} values, Figure~\ref{fig:gate_hist} reveals that the MLP exact-hit histogram is heavily dominated by the constant \texttt{TRUE} gate in this regime.
Thus, primitive recoverability should be interpreted jointly with the gate histogram: a high best-match score can reflect substantial mass on trivial (constant) gates rather than widespread recovery of non-constant fan-in-$2$ logic.
This motivates treating BNR (Table~\ref{tab:bnr_only}) as the primary neuronwise interpretability certificate, with primitive probes serving as auxiliary diagnostics.

\subsection{Expression complexity}
\label{app:expr_len}

We additionally report a coarse expression-complexity proxy based on token counts
(number of variable occurrences plus Boolean operators) of the \emph{label} formulas and the \emph{decoded \ModelNameShort} expressions.
Since dataset formulas are generated without algebraic simplification, shorter decoded expressions should be interpreted only as a simplicity bias among logically equivalent representations, not as a semantic advantage (Table~\ref{tab:expr_len}).

\begin{table}[t]
\centering
\small
\begin{tabular}{l c c}
\toprule
Setting & label tok & \ModelNameShort\ tok \\
\midrule
\ModelNameShort\ (main) & 11.80 $\pm$ 0.00 & 9.80 $\pm$ 0.15 \\
\bottomrule
\end{tabular}
\caption{Average token complexity of label expressions vs.\ decoded \ModelNameShort\ expressions (mean $\pm$ std across seeds). Label expressions are generated without algebraic simplification; decoded \ModelNameShort\ expressions are often shorter among logically equivalent forms.}
\label{tab:expr_len}
\end{table}

\subsection{How we selected the final hyperparameter setting (ablation path)}
\label{app:ablation_path}

We selected the final experimental configuration via a staged ablation, keeping the dataset and evaluation protocol fixed (full truth-table training and EM evaluation). The stages are:

\paragraph{Stage I: coarse hyperparameter grid (including lifting).}
We first ran a broad grid over routing strategy, repulsion settings, optimizer choice, and the lifting toggle (among other knobs), and ranked configurations by EM. Table~\ref{tab:top8_full_grid} and Table~\ref{tab:top8_lift_full_grid} reports the top configurations (aggregated across seeds), and motivates our choice of the baseline hyperparameter setting used in later stages. In this stage, lifting (\texttt{lifting.use}) is treated as a tunable option.

\begin{table}[t]
\centering
\small
\begin{tabular}{lllll c c}
\toprule
route & repel & mode & eta & lam\_const16 & EM & n \\
\midrule
mi\_soft & False & log & 2 & 0 & 0.8651 $\pm$ 0.0124 & 20 \\
mi\_soft & False & log & 2 & 0.005 & 0.8578 $\pm$ 0.0126 & 20 \\
mi\_soft & False & log & 2 & 0.001 & 0.8557 $\pm$ 0.0126 & 20 \\
mi\_soft & True & hard-log & 2 & 0.001 & 0.8435 $\pm$ 0.0006 & 19 \\
mi\_soft & True & log & 2 & 0.001 & 0.8435 $\pm$ 0.0006 & 20 \\
mi\_soft & True & log & 2 & 0.005 & 0.8435 $\pm$ 0.0004 & 18 \\
mi\_soft & True & hard-log & 2 & 0.005 & 0.8435 $\pm$ 0.0004 & 20 \\
mi\_soft & True & hard-log & 2 & 0 & 0.8432 $\pm$ 0.0012 & 20 \\
\bottomrule
\end{tabular}
\caption{Top-8 configurations on full\_grid (mean $\pm$ std across seeds).}
\label{tab:top8_full_grid}
\end{table}

\begin{table}[t]
\centering
\small
\begin{tabular}{lllll c c}
\toprule
route & repel & mode & eta & lam\_const16 & EM & n \\
\midrule
learned & True & hard-log & 2 & 0.005 & 0.7496 $\pm$ 0.0056 & 20 \\
learned & True & hard-log & 2 & 0 & 0.7457 $\pm$ 0.0055 & 18 \\
learned & True & hard-log & 2 & 0.001 & 0.7408 $\pm$ 0.0078 & 19 \\
learned & True & log & 2 & 0.005 & 0.7391 $\pm$ 0.0045 & 20 \\
learned & True & log & 2 & 0 & 0.7367 $\pm$ 0.0061 & 18 \\
learned & False & log & 2 & 0 & 0.7366 $\pm$ 0.0055 & 18 \\
learned & False & log & 2 & 0.005 & 0.7357 $\pm$ 0.0044 & 15 \\
learned & False & log & 2 & 0.001 & 0.7296 $\pm$ 0.0079 & 18 \\
\bottomrule
\end{tabular}
\caption{Top-8 configurations on lift\_full\_grid (mean $\pm$ std across seeds).}
\label{tab:top8_lift_full_grid}
\end{table}

\paragraph{Stage II: 16-gate interpolant ablation (holding the Stage-I setting fixed).}
Fixing the best-performing Stage-I hyperparameters, we ablated the differentiable 16-gate interpolant used during training:
\texttt{sigma16.mode} $\in\{\texttt{rbf},\texttt{bump},\texttt{lagrange}\}$.
Table~\ref{tab:sigma16_mode_ablation} summarizes the results and shows that \texttt{rbf} performs best in our regime; we therefore fix \texttt{sigma16.mode=rbf} for subsequent experiments.

\begin{table}[t]
\centering
\small
\begin{tabular}{l c c}
\toprule
$\sigma_{16}$ mode & EM & row-acc \\
\midrule
rbf & 0.8721 $\pm$ 0.0024 & 0.9908 $\pm$ 0.0004 \\
bump & 0.8477 $\pm$ 0.0044 & 0.9873 $\pm$ 0.0006 \\
lagrange & 0.5694 $\pm$ 0.0178 & 0.9552 $\pm$ 0.0022 \\
\bottomrule
\end{tabular}
\caption{Ablation over the $\sigma_{16}$ gate-interpolant choice on the \texttt{mode\_cmp} suite (mean $\pm$ std across seeds).}
\label{tab:sigma16_mode_ablation}
\end{table}

\paragraph{Stage III: width/depth budget sweep around dataset proxies.}
Finally, holding all other hyperparameters fixed (including \texttt{sigma16.mode=rbf}), we swept additional width/depth budget relative to dataset-provided proxies:
\[
S_{\mathrm{model}} = S_{\mathrm{base}} + S_{\mathrm{add}},
\qquad
L_{\mathrm{model}} = L_{\mathrm{base}} + L_{\mathrm{add}}.
\]
Figure~\ref{fig:sweep_em} shows that increasing $S_{\mathrm{add}}$ substantially improves EM, while increasing $L_{\mathrm{add}}$ tends to degrade EM in this regime. Based on this sweep we select $S_{\mathrm{add}}=10$ and $L_{\mathrm{add}}=0$ as the final setting used for the \ModelNameShort-vs-MLP comparisons in the main text.

\subsection{Width/depth sweep: effect of extra budget}
\label{app:sweep}

\begin{figure}[htpb]
\centering
\includegraphics[width=0.85\linewidth]{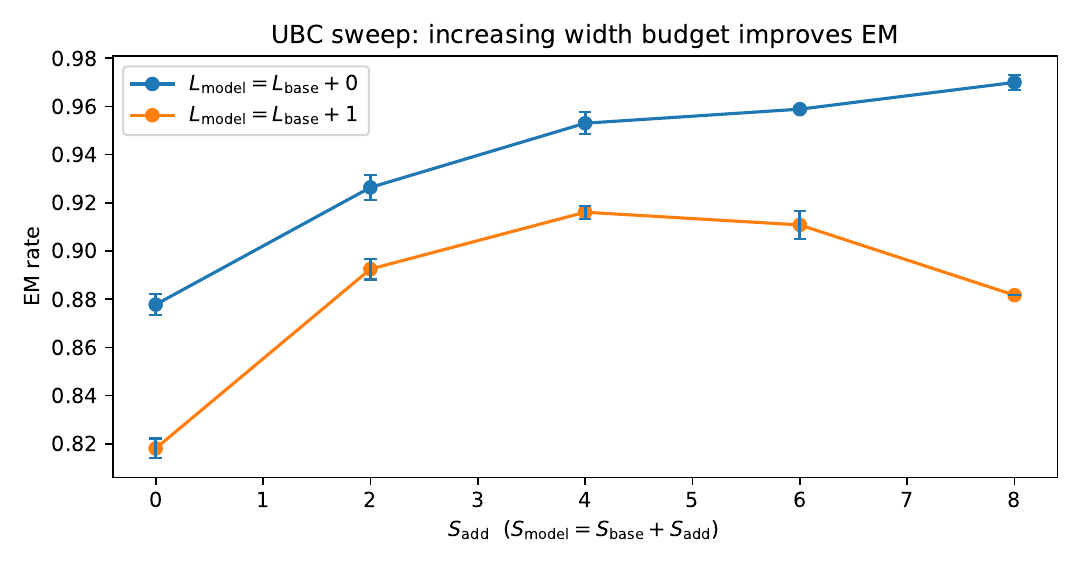}
\caption{\ModelNameShort\ sweep: EM vs.\ $S_{\mathrm{add}}$ for different $L_{\mathrm{add}}$, where $S_{\mathrm{model}}=S_{\mathrm{base}}+S_{\mathrm{add}}$ and $L_{\mathrm{model}}=L_{\mathrm{base}}+L_{\mathrm{add}}$.}
\label{fig:sweep_em}
\end{figure}

\section{Additional Related Work}
\label{app:rw}
\paragraph{Neuro-symbolic reasoning and probabilistic circuits.}
Neuro-symbolic methods integrate neural components with symbolic structure or logic-based inference.  Logic Tensor Networks enforce differentiable satisfaction of first-order constraints during training~\cite{Badreddine2022LTN}, while neural probabilistic logic programming systems such as DeepProbLog embed neural predicates inside probabilistic programs for joint learning and inference~\cite{manhaeve2018deepproblog}.  Probabilistic circuits, including sum-product networks (SPNs), represent distributions as structured compositions that admit efficient exact inference under tractable structural conditions~\cite{PoonDomingos2011SPN,ProbCirc20}.  Our approach is closest in spirit in that it builds in an intrinsic circuit-level representation rather than extracting one post hoc; we specialize to fan-in-$2$, fan-out-$1$ Boolean circuits and obtain an explicit certificate of neuronwise Booleanity.  Extending the construction to richer circuit families (e.g., variable fan-in/fan-out or alternative gate sets) is a natural direction for future work.

\end{document}